\documentclass[10pt,twocolumn]{IEEEtran} 	

\usepackage{fancyhdr, epsfig, epsf, amsthm, amsmath, amssymb, amsfonts, subfigure, color}
\usepackage{threeparttable}
\usepackage[noadjust]{cite}
\usepackage{dsfont}
\usepackage{enumerate}
\usepackage{comment}
\usepackage{soul}
\usepackage{algorithm}
\usepackage{algpseudocode}
\usepackage{mathtools}
\usepackage{textcomp}
\usepackage{gensymb}
\usepackage{array,booktabs,hhline,tabu}
\newcolumntype{P}[1]{>{\centering\arraybackslash}p{#1}}
\usepackage{bbm}
\usepackage{amssymb}
\usepackage{multirow}
\usepackage{adjustbox}
\usepackage{hyperref} 
\usepackage{cleveref}
\Crefname{figure}{Fig.}{Figs.}
\usepackage{scalerel}
\usepackage{tikz}
\usetikzlibrary{svg.path}

\definecolor{orcidlogocol}{HTML}{A6CE39}
\tikzset{
	orcidlogo/.pic={
		\fill[orcidlogocol] svg{M256,128c0,70.7-57.3,128-128,128C57.3,256,0,198.7,0,128C0,57.3,57.3,0,128,0C198.7,0,256,57.3,256,128z};
		\fill[white] svg{M86.3,186.2H70.9V79.1h15.4v48.4V186.2z}
		svg{M108.9,79.1h41.6c39.6,0,57,28.3,57,53.6c0,27.5-21.5,53.6-56.8,53.6h-41.8V79.1z M124.3,172.4h24.5c34.9,0,42.9-26.5,42.9-39.7c0-21.5-13.7-39.7-43.7-39.7h-23.7V172.4z}
		svg{M88.7,56.8c0,5.5-4.5,10.1-10.1,10.1c-5.6,0-10.1-4.6-10.1-10.1c0-5.6,4.5-10.1,10.1-10.1C84.2,46.7,88.7,51.3,88.7,56.8z};
	}
}

\hypersetup{
	colorlinks=false
}

\newcommand\orcidicon[1]{\href{https://orcid.org/#1}{\mbox{\scalerel*{
				\begin{tikzpicture}[yscale=-1,transform shape]
					\pic{orcidlogo};
				\end{tikzpicture}
			}{|}}}}

\newtheorem{proposition}{Proposition}

\newtheorem{assumption}{Assumption}

\newtheorem{remark}{\bf Remark}
\def\phi{\varphi}

\def\l{\left}
\def\r{\right}
\def\({\left(}
\def\){\right)}

\def\b0{{\mathbf{0}}}

\includecomment{comment}	

\DeclarePairedDelimiter\norm{\lVert}{\rVert}

\begin{document}
\title{Exploiting User Mobility for WiFi RTT Positioning: A Geometric Approach}
\author{Kyuwon Han \orcidicon{0000-0001-5792-8414}, \IEEEmembership{Student Member, IEEE}, Seung Min Yu, Seong-Lyun Kim \orcidicon{0000-0002-5228-9913}, \IEEEmembership{Member, IEEE}, and Seung-Woo Ko \orcidicon{0000-0002-8592-7408}, \IEEEmembership{Member, IEEE} \\
\thanks{ 	
	This work was supported by Institute of Information \& communications Technology Planning \& Evaluation (IITP) grant funded by the Korea government(MSIT) (No. 2018-11-1864, Scalable Spectrum Sharing for Beyond 5G Communication), a grant from R\&D Program of the Korea Railroad Research Institute, Republic of Korea, and in part by the Basic Science Research Program through the National Research Foundation of Korea (NRF) funded by the Ministry of Science and ICT (NRF-2019R1G1A1100123).
\textit{(Corresponding author: Seung-Woo Ko.)}
	
K. Han and S.-L Kim are with Yonsei University, Seoul, South Korea (email: \{kwhan,slkim\}@ramo.yonsei.ac.kr). S. M. Yu is with the Korea Railroad Research Institute, Uiwang, South Korea
(email: smyu@krri.re.kr). S.-W. Ko is with the Division of Electronics and Electrical Information Engineering, Korea Maritime and Ocean University (KMOU), Busan 49112, South Korea (email: swko@kmou.ac.kr). 

Copyright (c) 2021 IEEE. Personal use of this material is permitted. However, permission to use this material for any other purposes must be obtained from the IEEE by sending a request to pubs-permissions@ieee.org.
}}

\IEEEoverridecommandlockouts 
	
\markboth{ACCEPTED FOR PUBLICATION IN IEEE INTERNET OF THINGS JOURNAL}
{}
\maketitle

\begin{abstract}

Recently, \emph{round-trip time} (RTT) measured by a fine-timing measurement protocol has received great attention in the area of WiFi positioning. It provides an acceptable ranging accuracy in favorable environments when a \emph{line-of-sight} (LOS) path exists. Otherwise, a signal is detoured along with non-LOS paths, making the resultant ranging results different from the ground-truth, called an RTT bias, which is the main reason for poor positioning performance. To address it, we aim at leveraging the user mobility trajectory detected by a smartphone's inertial measurement units, called \emph{pedestrian dead reckoning} (PDR). Specifically, PDR provides the geographic relation among adjacent locations, guiding the resultant positioning estimates' sequence not to deviate from the user trajectory. To this end, we describe their relations as multiple geometric equations, enabling us to render a novel positioning algorithm with acceptable accuracy. Depending on the mobility pattern being linear or arbitrary, we develop different algorithms divided into two phases. First, we can jointly estimate an RTT bias of each AP and the user's step length by leveraging the geometric relation mentioned above. It enables us to construct a user's relative trajectory defined on the concerned AP's local coordinate system. Second, we align every AP's relative trajectory into a single one, called \emph{trajectory alignment}, equivalent to transformation to the global coordinate system. As a result, we can estimate the sequence of the user's absolute locations from the aligned trajectory. Various field experiments extensively verify the proposed algorithm's effectiveness that the average positioning error is approximately 0.369 (m) and 1.705 (m) in LOS and NLOS environments, respectively.
\end{abstract}


\section{Introduction}\label{Sec:Intro}

Estimating  a user's  location, called \emph{positioning}, has become vital as a clue to assimilate his behaviors and predict the demands to realize the vision of smart cities \cite{BookAlan2016, Laoudias2018}. Along with the appearance of smartphones equipping various built-in sensors and communication modules, positioning has become an attractive research area due to its promising potential capable of combining various techniques. Among them, this work focuses on \emph{round-trip times} (RTTs) from multiple WiFi \emph{access points} (APs) and a user's mobility pattern measured by the built-in sensors, each of which has a different view from the positioning perspective. A WiFi RTT-based approach gives a macroscopic view of the region where the user is likely to be located during the movement.  On the other hand, the user movement pattern provides a microscopic view of how the user is moving at a specific instant. As a result, combining the two renders the relation between point estimates, forming a geometric representation concerning the given measurements. It enables the design of a novel  algorithm to achieve accurate positioning. 

\subsection{Prior Works}  \label{Subsec:Intro_PriorWorks}

\subsubsection{WiFi positioning} \label{Subsubsec:Intro_WiFiPositioning}

As the massive number of WiFi APs have been deployed in our surroundings, WiFi positioning has received significant attention to providing users seamless positioning services anytime and anywhere \cite{Yang2015}. Depending on the means used to capture specific physical properties of radio signals, two kinds of approaches exist in the literature: \emph{received signal strength} (RSS)-based and RTT-based approaches. The former uses a mathematical path-loss model to translate the measured RSS to the distance to a WiFi AP. Given the {ranging results} to at least $3$ APs, the user's position is uniquely estimated using a multilateration method.  
RSS has been widely used as a fundamental positioning element due to its simple accessibility that RSS is measurable in a commercial off-the-shelf Android smartphone \cite{Vy2019}.  However, the randomness of a radio signal, e.g., shadowing and short-term fading, makes it challenging to obtain reliable ranging results, especially in indoor environments where numerous reflectors and blockages exist. On the other hand, RTT is a relatively reliable measurement since it exploits one basic theory of {classical} physics that the speed of signal propagation is constant at light speed   $c=3\cdot 10^8$ (m/sec). The measurement of RTT is enabled by \emph{fine timing measurement} (FTM) protocol, which was firstly introduced in IEEE 802.11 mc \cite{IEEE2016}. After several {handshaking signals} between a smartphone and the paired AP, the timing instants of the signal receptions are shared, enabling the computation of RTT. With several super-resolution methods, FTM returns a precise RTT {estimate} at picosecond granularity in favorable environments, i.e., a \emph{line-of-sight} (LOS) scenario \cite{Banin2017}.  

Nevertheless, complicated surroundings with numerous reflectors and blockages make the concerned radio signal detour {differently} from a LOS path, called a \emph{{non}-LOS} (NLOS)  path.  The resultant ranging {results} is thus \emph{biased}, as mentioned in \cite{Ibrahim2018}, which is the main reason behind the performance limit. Several methods have been suggested in the literature to overcome the limitation. A primary way is to identify whether the observed signal path is LOS or NLOS and calibrate the bias.  In \cite{Si2020},  the likelihood of a LOS  path is derived, following the assumption that RSS is a Gaussian random variable. The recent trend of machine learning enables the LOS/NLOS identification without any statistical hypothesis. 
For example, in \cite{Han2019} and \cite{Bregar2018},  LOS and NLOS paths are identified by supervised learning techniques, e.g., support vector machine and artificial neural network respectively, of which the performances depend on the number of labeled data.  In \cite{Choi2019}, the technique of unsupervised learning is used to infer the bias by designing a novel cost function underlying the fact that the spatial information, e.g., location, distance, and velocity, is temporally correlated. However, all machine learning-based techniques mentioned above need an offline phase such that support vectors or neural networks should be trained in advance concerning all possible sites, making their usability and scalability limited.   \\

\subsubsection{Estimating User Mobility}\label{Subsubsec:Intro_PDR}

Understanding user mobility gives significant benefits for seamless positioning services by enabling a user to estimate his location in \emph{global positioning system} (GPS)-restricted areas, i.e., {inside buildings} or tunnels. A user's current location can be updated from the latest GPS signal by integrating the trajectory of user movement, called \emph{dead reckoning} (DR). The performance of DR relies on the accuracy of estimating a {user's} mobility, and DR works efficiently for vehicular positioning since its on-board sensors, called \emph{inertial measurement units} (IMUs), can accurately measure velocity, orientation, {etc}~\cite{Toledo2009}.

Noting that modern smartphones also possess IMUs, the concept of DR can be applied to the positioning using a smartphone, defined as \emph{pedestrian DR} (PDR) \cite{Yang2015_PDR_survey}. Basically, the IMUs of a typical smartphone comprise {accelerometer}, gyroscope, and {magnetometer}, each of which observes a user's mobility from a different aspect. An accelerometer can count a user's number of walking steps from the repetitive patterns of up-and-down accelerations, translated into the corresponding total moving distance by multiplying his step length. Next, a gyroscope can recognize {turning direction to right or left sides} when its measurement is suddenly changed. Last, a magnetometer can detect a heading direction in favorable conditions such as outdoor-like spaces  \cite{Afzal2011}. Combining them leads to constructing the user's full trajectory. 

Despite its advantages mentioned above, {the effectiveness of PDR as a standalone positioning technique is questionable due to the following reasons.} First, a user's step length should be known in advance as a prerequisite to translate the number of steps into the distance. Several formulas have been proposed in the literature to estimate it without the user's direct input (see,  e.g., \cite{Weinberg2002} and \cite{Kang2015}), most of which are designed based on the rule of a thumb. 
On the other hand, it should reflect the user's characteristics such as height, weight, and speed and acceleration of walking, hindering the generalization into a simple formula.  
Second, a heading direction estimated by a magnetometer has a significant offset, affected by a magnetic distortion due to several indoor materials and a misalignment between the smartphone's heading direction and the real moving direction \cite{Scherzinger1996}. Third, the concerned scenarios of PDR are mostly indoor, where a GPS signal is hardly detected. In other words, its accuracy cannot be guaranteed and deteriorates as time passes due to the accumulation of estimation errors.  To cope with the above issues, it is recommended to incorporate PDR into other positioning systems by {providing additional location information} to calibrate these measurements.  The recent advancement of this area can be found in numerous surveys, such as~\cite{Guo2020}.  
\\

\subsubsection{Integrating WiFi Positioning and PDR}\label{Subsubsec:Intro_Integration}

The limitations of WiFi RTT positioning and PDR mentioned above are complementary and can be overcome by their integration, which is the main theme of this work. On the one hand,  the IMUs of a smartphone excluding a magnetometer work independently of surrounding environments, playing a pivotal role in compensating an environment-dependent RTT bias. On the other hand, the positions estimated by WiFi positioning can fill the missing information to complete the user's trajectory. Motivated by the synergy effect, a few recent works have been studied in the literature, most of which rely on a technique of an \emph{extended Kalman filter} (EKF). EKF is a well-known nonlinear state estimator utilizing a series of sequential observations with measurement noises.  In \cite{Choi2020}, for example, WiFi positioning's RTT bias, PDR's step length, and heading direction are jointly calibrated using EKF by inputting the raw measurements of PDR and WiFi RTTs. A similar approach is made in \cite{Sun2020}, {where} the positions estimated by PDR and the distances converted from WiFi RTTs are used as inputs of EKF. In \cite{Yu2020}, the PDR's measurements are pre-calibrated by EKF. The result is then fused with the estimated distances from WiFi RTTs to enable 3D positioning by an unscented particle filter, another state estimation filter. In \cite{Liu2021}, EKF is used to detect and remove WiFi RTT measurements' outliers. 
It is shown that all approaches mentioned above  can provide more accurate positioning results than standalone techniques. On the other hand, the convergence of EKF is not guaranteed and sometimes diverges due to the lack of statistical knowledge of measurement noises and the linearization of nonlinear functions required to compute the inputs' covariance matrices. The resultant location {estimates} can be inconsistent depending on the selection of initial settings, calling for diversifying positioning approaches other than EKF.

\begin{table*}[]
	\caption{{Summary of Techniques Integrating WiFi \& PDR }} \label{Table:PriorWorks}
	\begin{adjustbox}{width=\textwidth}
	\begin{tabular}{l|l|c|cc|cc|cc}
		\toprule
		\multirow{2}{*}{Approach}                                                               & \multirow{2}{*}{Key references} & \multirow{2}{*}{\begin{tabular}[c]{@{}c@{}}Ranging \\ calibration \end{tabular}} & \multicolumn{2}{c|}{Integration method} & \multicolumn{2}{c|}{Step length} & \multicolumn{2}{c}{Heading direction estimation} \\ \cline{4-9} 
		
		&                                                  &                                       & EKF-based     & Geometry-based     & Given          & Unknown         & W/ magnetometer           & W/o magnetometer     \\ \hline
		WiFi RTT Positioning                                                                    
		& \cite{Si2020, Han2019,Banin2016a, Hsieh2019} & \checkmark  &    &    &   &    &    &     \\ \hline
		\multirow{2}{*}{\begin{tabular}[c]{@{}l@{}}WiFi RSS \& PDR \\ integration\end{tabular}} 
		& \cite{Vy2019}            &            &            &            &            & \checkmark & \checkmark &          \\ \cline{2-9} 
		& \cite{Choi2020b}         & \checkmark &            &            & \checkmark &            &            & \checkmark  \\ \hline
		\multirow{4}{*}{\begin{tabular}[c]{@{}l@{}}WiFi RTT \& PDR \\ integration\end{tabular}} 
		& \cite{Sun2020,  Liu2021} & \checkmark & \checkmark &            & \checkmark &            & \checkmark &         \\ \cline{2-9} 
		& \cite{Yu2020,Xu2019}     & \checkmark &            &            & \checkmark &            & \checkmark &         \\ \cline{2-9} 
		& \cite{Choi2020}          & \checkmark & \checkmark &            &            & \checkmark &            & \checkmark \\ \cline{2-9} 
		& Proposed                 & \checkmark &            & \checkmark &            & \checkmark &            & \checkmark  \\ \bottomrule
	\end{tabular}
	\end{adjustbox}
\end{table*}

\subsection{Main Contributions}

{This work aims} at designing a new hybrid positioning design combining WiFi RTTs and PDR measurements without EKF. To this end, we attempt to adopt a geometric approach describing the relationship between multiple measurements in mathematical form. This approach provides twofold benefits from the positioning perspective. First, the relations mentioned above are summarized as a list of equations, forming a \emph{system of equations} (SOE) whose unknowns are related to RTT biases, a user's step length, heading direction, {etc}. The SOE can be iteratively solved using well-known optimization techniques, guaranteeing the convergence to a local optimal, and possible to solve it by a single matrix inversion if {the SOE is linear}. Second, we can rigorously provide requirements to guarantee the uniqueness and existence of the positioning result, e.g., the minimum numbers of detected walking steps and connected APs, which are equivalent to the conditions for unique positioning. It helps design not only a practical positioning algorithm, always returning an accurate position but also efficient WiFi deployments in a concerned area.   

There have been several works adopting a geometric approach in different systems such as cellular positioning \cite{Miao2007}, RADAR \cite{Frischen2017}, and vehicular positioning \cite{Han2019a}. Furthermore, during the revision of our paper, a few recent works integrating WiFi positioning and PDR without EKF have been observed, based on using another filter \cite{Xu2019} and deep neural network \cite{Choi2020b}. On the other hand, none of the works uses a geometric approach. To the best of our knowledge, this work represents the first attempt to design a geometry-based positioning algorithm concerning the integration between FTM and PDR. We summarize the comparison with prior works in \Cref{Table:PriorWorks}.
The  main contributions are summarized~below.


\begin{itemize}
	\item \textbf{Joint estimation of an RTT bias and a step length}: We aim at enhancing  an RTT-based distance estimation by offsetting the RTT bias, coupled with PDR information like the number of walking steps, direction changes, and step length. Noting that all information  excluding step length can be measurable by PDR, an SOE is formed to jointly estimate two unknowns of RTT bias and step length. In the linear mobility case, the SOE can be transformed into a linear structure solvable by a matrix inversion when at least $4$ steps are detected. In the arbitrary mobility case, a \emph{one-dimensional} (1D) search enables us to solve the SOE, guaranteeing its unique solution if at least $5$ steps are detected. Besides, the positioning error due to measurement noises can be significantly reduced by leveraging diversified positioning results through multiple SOEs formed by different step combinations. 
	\item \textbf{New positioning method using a trajectory alignment}: The user's sequential positions during his movement can be estimated by aligning multiple trajectories derived from the estimated RTT bias of each WiFi AP and step length. It is equivalent to find the user's initial heading direction, another information challenging to be estimated using a smartphone, as aforementioned. The minimum number of WiFi APs required to {finding} the user's position uniquely is $3$ or $2$ depending on the linear or arbitrary mobility pattern, respectively. Given the requirement of the minimum number of APs and with measurement noises, we develop a 1D search-based  algorithm to make all trajectories aligned as closely as possible.  
	\item \textbf{Verification by field experiments}: The proposed positioning algorithms are evaluated based on field experiments and found to be effective. Specifically, the resultant {average} positioning error is reduced to $1.71$ (m) in unfavorable environments like an underground parking lot, where numerous obstacles exist, such as vehicles, walls, and pillars.  
\end{itemize}

The remainder of the paper is organized as follows. Section \ref{Sec:SystemModel} introduces the system model, measurement procedures, and the overview of the entire positioning algorithm. Section \ref{Section:RTTRangingEnhancement} presents a technique jointly estimating RTT bias and step length for a single WiFi AP. 
Given the enhanced ranging results of multiple WiFi APs, the positioning algorithm finding the user's location is developed  
in Section \ref{Sec:TrajectoryAlignment}. Experiment results are presented in Section \ref{Sec:Experiement}, followed by concluding remarks in Section \ref{Sec:Conclusion}.

\section{System Model}\label{Sec:SystemModel}

Consider the scenario comprising one user with a smartphone and $M$ WiFi APs, denoted by a set $\mathbb{M}=\{1,\cdots M\}$, as shown in Fig. \ref{Fig:SystemModel}. The smartphone has built-in IMU sensors and a WiFi module  required to measure the user's mobility pattern and RTTs, respectively. The detailed measurement procedures are firstly explained. 
Next, the geometric relations between the measurements and the user's locations are derived. Last, the overview of the proposed approach is described. 

\begin{figure}[t]
	\centering
	\includegraphics[width=8cm]{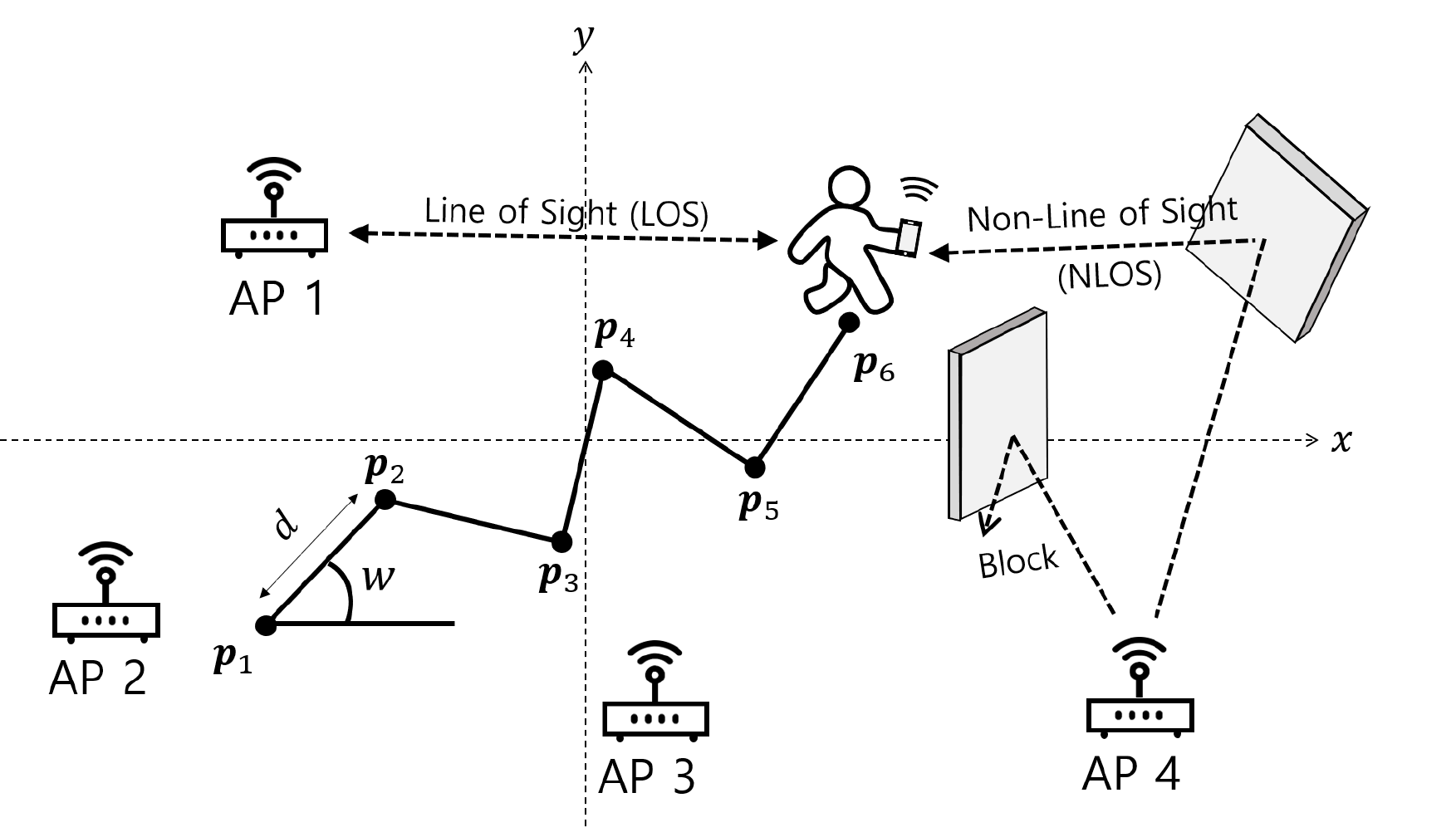}
	\caption{The scenarios illustration with one walking user holding a smartphone and multiple WiFi APs ($M=4$).  Depending on the user's location, a different propagation condition to each AP is determined between LOS and NLOS, but the user does not have any prior knowledge related to it.   }
	\label{Fig:SystemModel}
\end{figure}

\subsection{Measurements}
This subsection explains the measurement procedures and the corresponding outputs required to form geometric relations in the next subsection. 

\subsubsection{Mobility Pattern}
Among IMU sensors in the smartphone, we use an accelerometer and a gyroscope to detect walking steps and turning directions, respectively\footnote{This work does not use a magnetometer due to its distortion as stated in Sec.~\ref{Subsubsec:Intro_PDR}. Incorporating the measurement of a magnetometer with advanced calibration techniques such as \cite{Malyugina2014} and \cite{Vasconcelos2011} is interesting, deserving further investigation in the future. }. To be specific, the accelerometer's up-and-down acceleration is referred to as the user's one walking step. Consider that $N$ walking steps are detected, each of which the index is $n\in\mathbb{N}$, where $\mathbb{N}=\{1,\cdots N\}$. 
The instant of detecting the $n$-th step is denoted by $t_n$. 
When detecting a walking step at $t_n$, the smartphone checks its gyroscope between $t_n$ and $t_{n+1}$, whether the user's movement direction is changed or not. Denote $\mu_n$ the change of his moving direction when the $n$-th step is detected. We set $\mu_n=0$ unless the direction is changed. Besides, denote $\theta_n$ the accumulate direction change by the $n$-th step, namely,  $\theta_n=\sum_{k=1}^n \mu_k$. The initial values of $\mu_1$ and $\theta_1$ are zero without loss of generality. All measurements mentioned above are illustrated in Fig.~\ref{Fig:Moving}(a).


\subsubsection{RTT} The smartphone's WiFi modules and all WiFi APs support IEEE 802.11 mc or later specifications, enabling the measurement of RTTs between them via the FTM protocol. The FTM protocol is initiated whenever a walking step is detected at time $t_n$. Denote $\boldsymbol{\tau}_n=\left[\tau_n^{(1)},\cdots, \tau_n^{(M)}\right]$ the vector of the measured RTTs corresponding to the $n$-th walking step, where $\tau_n^{(m)}$ represents the RTT measurement from WiFi AP $m$ at $t_n$. 

Noting that it takes negligible time to complete an FTM procedure (approximately $30$~ms according to \cite{Banin2017}), it is reasonable to assume that all measurements concerning the $n$-th walking step are synchronized. As a result, we group them as a set of measurement denoted by $\mathcal{M}_n=\{t_n, \theta_n, \boldsymbol{\tau}_n\}$, $\forall n\in \mathbb{N}$. For ease of exposition, the parts described from now on are assumed to be free from measurement errors unless specified.  
However, the proposed algorithms explained in the sequel are designed to be able to work well in the presence of measurement errors.

\subsection{Relations between the Measurements}

\addtocounter{footnote}{-1}
\begin{figure}[t]
	\centering
		\subfigure[ The accelerometer and gyroscope's measurements]{\includegraphics[width=8cm]{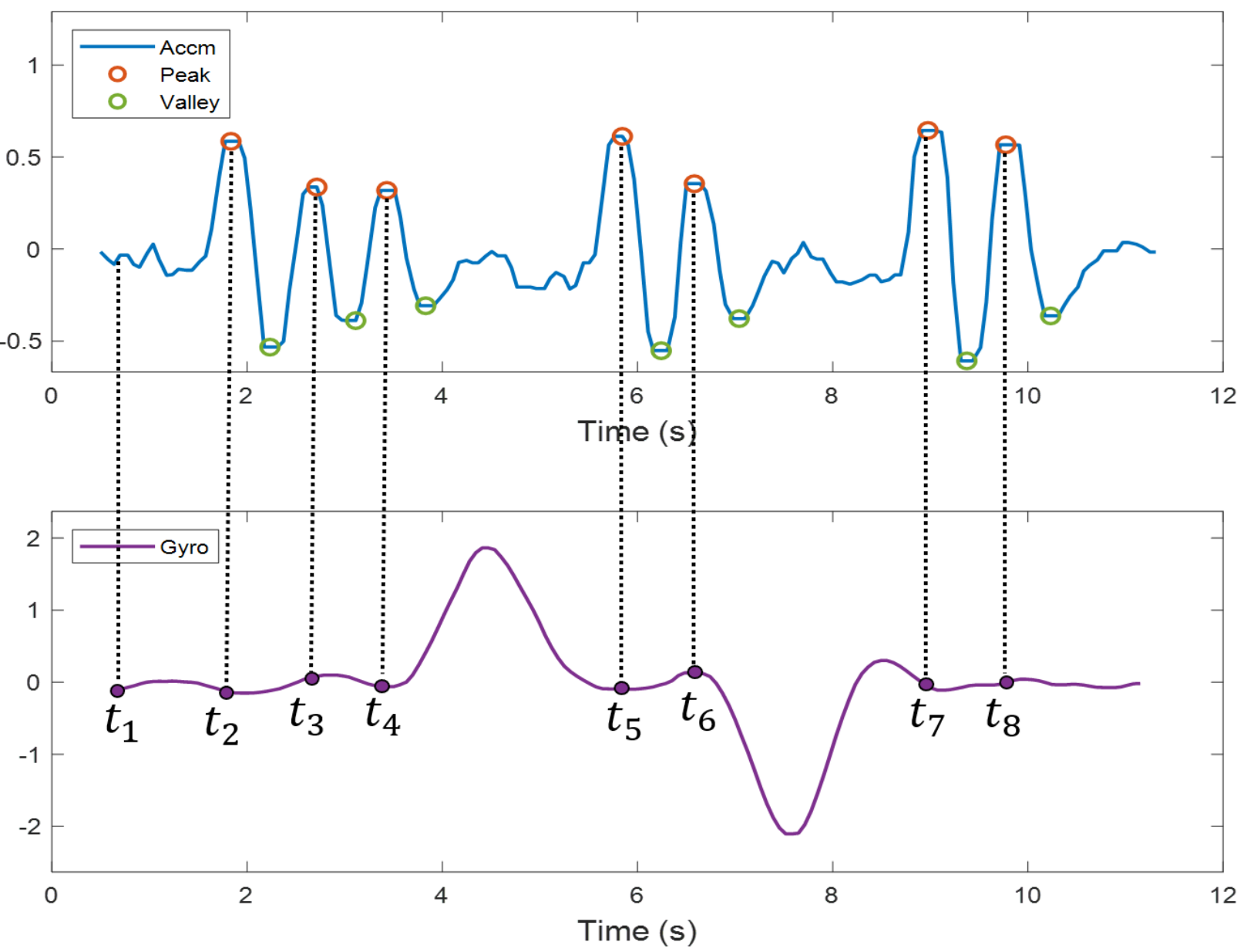}}
		\subfigure[{The user's trajectory}]{\includegraphics[width=8cm]{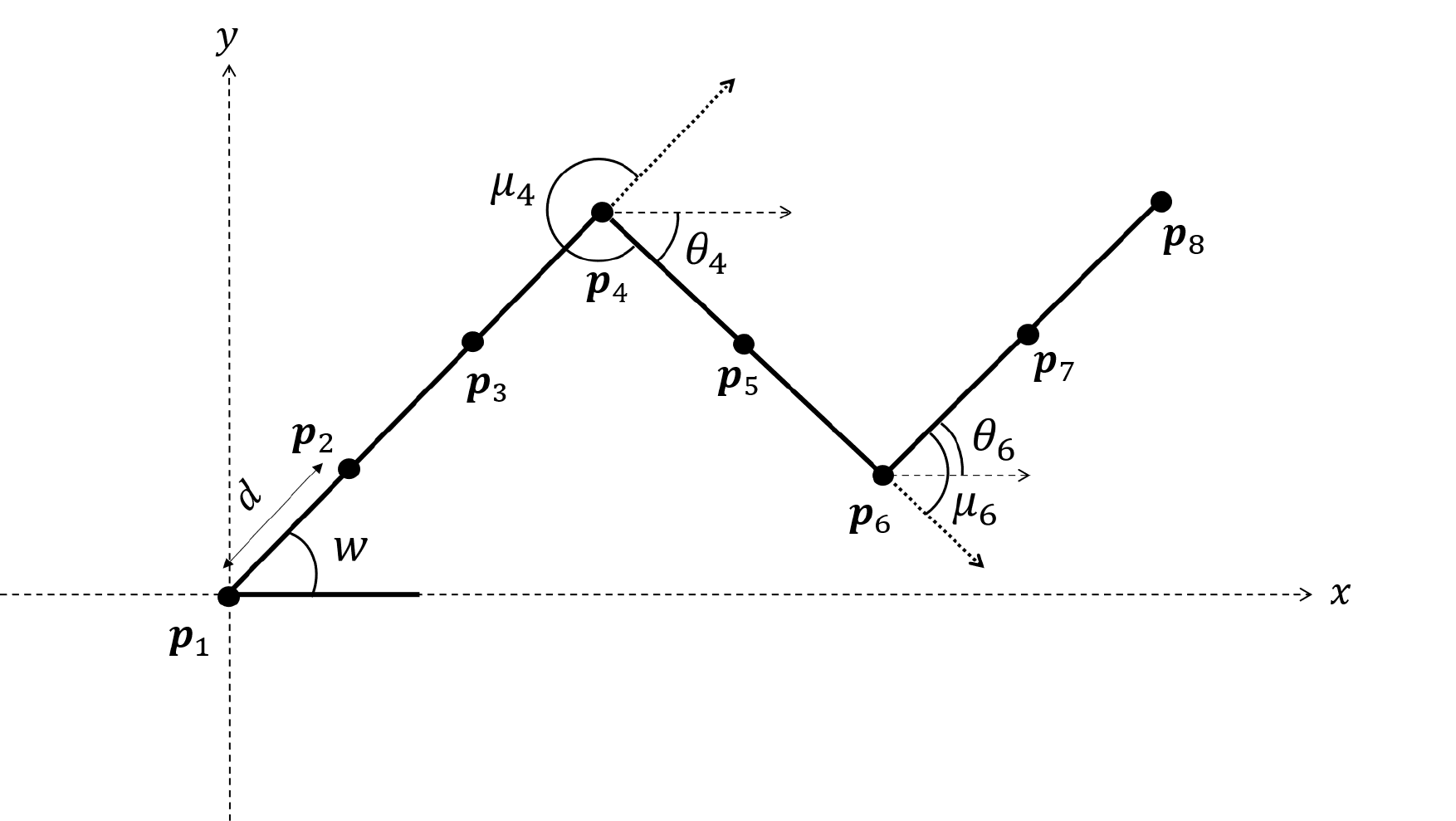}}
	\caption{The graphical example of  the accelerometer and gyroscope's measurements and the corresponding mobility pattern \protect\footnotemark. The accelerometer's peak point is recorded as the instant of each walking step. The direction change can be detected from the gyroscope's change between adjacent instants. }
	\label{Fig:Moving}
\end{figure}
\footnotetext{{It is implicitly assumed that accelerometer and gyroscope's measurements are relatively accurate and consistent by fixing the smartphone's location. Otherwise, severe measurement errors will be introduced, resulting in inaccurate estimations of step counting and direction changes.  It is interesting to extend the current design without the assumption by adopting the well-known pose-detecting algorithms presented in the literature, e.g., \cite{Wang2018b} and \cite{Lee2019}, outside the scope of current work.}}

We first explain the sequential change of the user's location due to user mobility.  Consider a 2D global Cartesian coordinate system, as illustrated in Fig. \ref{Fig:Moving}(b). 
Denote $\boldsymbol{p}_n=[x_n, y_n]^T$ the coordinates of the user's location concerning to the measurement set $\mathcal{M}_n$. The relation between  which is given~as
\begin{align} \label{eq:temporal_relation}
\boldsymbol{p}_{n+1} = \boldsymbol{p}_{n}+ d
\begin{bmatrix}
\cos(\omega+\theta_{n}), \sin(\omega+\theta_{n})
\end{bmatrix}^T,
\end{align}
where $\omega$ is a heading direction when the first walking step is detected at $t_1$. The variable $d$ represents the user's step length assumed to be constant, reflecting on a typical user's regular walking pattern\footnote{{It is verified from \cite{Jasuja1997} and \cite{Menz2003} that the assumption of constant step length is valid when a user’s activity pattern is unchanged, e.g., walking, running, and sitting. As a result, monitoring the user’s activity pattern using well-known methods, such as \cite{Bisio2012} and \cite{Zhou2019}, enables to seamlessly use the proposed algorithm by dividing the full trajectory into several sub-trajectories depending on the activity pattern.}}. Both $\omega$ and $d$ are unknowns to be estimated. 

Second, the relation of the RTT $\tau_n^{(m)}$ to the user's location $\boldsymbol{p}_n$ is described as follows. The RTT $\tau_n^{(m)}$ can be translated into the propagation distance by multiplying $\frac{c}{2}$ where $c$ is light speed. It is in general biased and larger than the direct distance  
$\parallel\boldsymbol{p}_{\text{AP}}^{(m)}-\boldsymbol{p}_{n}\parallel$, where $\boldsymbol{p}_{\text{AP}}^{(m)}=[x_{\text{AP}}^{(m)}, y_{\text{AP}}^{(m)}]^T$ is AP $m$'s coordinates and $\parallel\cdot \parallel$ represents the Euclidean distance. The relation between the two is thus given as
\begin{align} 
\frac{c\cdot \tau_n^{(m)}}{2}=\parallel \boldsymbol{p}_{\text{AP}}^{(m)}-\boldsymbol{p}_{n} \parallel+b,\quad \forall n\in\mathbb{N}, \quad \forall m\in\mathbb{M}, \nonumber
\end{align}
where $b$ represents the AP $m$'s RTT bias. The bias $b$ is a random variable drawn from an unknown stochastic distribution, since it is affected by various factors such as {carrier frequency, bandwidth, and surrounding materials \cite{Horn2020}}, and it is too complex to derive the distribution in a tractable form. The randomness of $b$ is one dominant reason hampering accurate positioning. On the other hand, we regard $b$ as a deterministic value as stated in the following assumption, which helps design a tractable algorithm in the sequel. 

\begin{assumption}[Deterministic Bias]\label{Assumption:Bias}\emph{The bias experienced from the same WiFi AP is assumed to be constant but unknown, denoted by $b^{(m)}$ for all $m\in\mathbb{M}$.}
\end{assumption}

\begin{remark}[Effect of Deterministic Bias]\emph{
{In general, RTT bias is not deterministic but random in a real environment, a key factor making indoor WiFi positioning challenging. Instead of directly addressing the randomness, we use the assumption of deterministic bias as a trick, allowing us to design a geometric positioning algorithm introduced in the sequel. It is worth highlighting that  the error due to the deterministic bias assumption can be reduced as marginal using two processes in the proposed algorithm, namely, reference step selection and their combination approach explained in Sec.~\ref{Section:RTTRangingEnhancement}. Their effect is well-explained with relevant experimental results in Appendix \ref{Subsec:EffectofBias}. }
}

\end{remark}
As a result, the above is rewritten as  
\begin{align}\label{Eq:DistanceRelation}
\parallel \boldsymbol{p}_{\text{AP}}^{(m)}-\boldsymbol{p}_{n} \parallel+b^{(m)}&=
 \frac{c\cdot \tau_n^{(m)}}{2}\nonumber\\
 &=r_n^{(m)},\quad \forall n\in\mathbb{N}, \quad \forall m\in\mathbb{M},
\end{align}
where $r_n^{(m)}$ represents the propagation distance directly obtained from the RTT $\tau_n^{(m)}$.    

\subsection{Procedure Overview}

\begin{figure}[t]
	\centering
	\includegraphics[width=9cm]{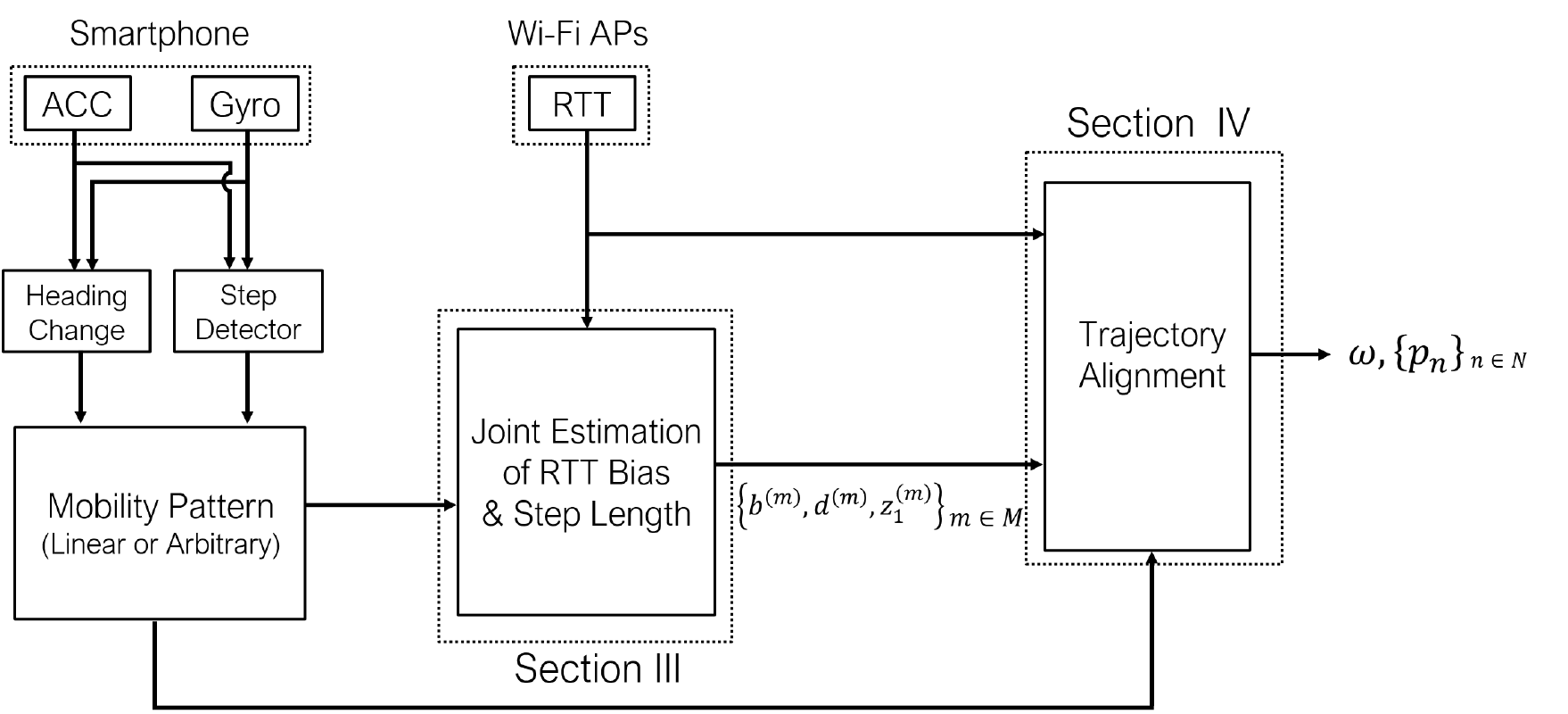}
	\caption{{Overview of the proposed algorithm. The step length $d^{(m)}$ and RTT bias $b^{(m)}$, initial relative position $z_1^{(m)}$, initial heading direction $\omega$, and user’s position $\{p_n\}$ are variables we estimate throughout the paper, specified in (\ref{Eq:Optimal_Solution_Median}), (\ref{Eq:Optimal_Solution_Median2}), (\ref{Eq:TAErrorFunction}), and (\ref{Eq:GlobalPosition}), respectively.}} 
	\label{Fig:Procedure} 
\end{figure}

{This subsection previews the entire procedure as shown in Fig. \ref{Fig:Procedure}. The user's locations $\{\boldsymbol{p}_n\}$ are estimate using the following two-stage approach. }

\subsubsection{Joint Bias-and-Step Length Estimation}
First, the RTT bias of each AP and the user's step length are jointly estimated. Let us explain the case of AP $m$ as an example. 
Collect all RTTs measured from AP $m$ and the user's direction changes, denoted by $\boldsymbol{\tau}^{(m)}=[{\tau}_1^{(m)}, \cdots, {\tau}_N^{(m)}]$ and $\boldsymbol\theta=[\theta_1,\cdots,\theta_N]$, respectively. Using $\boldsymbol{\tau}^{(m)}$ and  $\boldsymbol\theta$, the AP's deterministic bias $b^{(m)}$ and the user's step length, say $d^{(m)}$, can be founded by solving SOE derived from the geometric relations of \eqref{eq:temporal_relation} and \eqref{Eq:DistanceRelation}. The detailed procedure is elaborated in Section \ref{Section:RTTRangingEnhancement}.  
\subsubsection{User Location Estimation} 
Second, the sequence of the user's locations $\{\boldsymbol{p}_n\}$ are estimated using the above estimations of $d^{(m)}$ and $b^{(m)}$ that provide the user's relative locations $\{\boldsymbol{z}_n^{(m)}\}$ from AP $m$. All relative locations can be aligned into single ones using a new positioning method, corresponding to the user's real locations $\{\boldsymbol{p}_n\}$. The detailed procedure is elaborated in Section~\ref{Sec:TrajectoryAlignment}.

\section{Enhancing RTT-Based Ranging via Joint Bias \& Step Length Estimation}\label{Section:RTTRangingEnhancement}

In this section, we aim {at improving RTT-based ranging performance} by jointly estimating the RTT bias and step length.
First, a coordinate system is transformed to facilitate algorithm designs. Next, joint bias-and-step length estimation algorithms are developed for cases of linear and arbitrary mobilities. Last, a new algorithm based on multiple combinations of steps is proposed to reduce  performance degradation due to measurement noises. 

\begin{figure}[t]
	\centering
	\subfigure[Linear mobility]{\includegraphics[width=7cm]{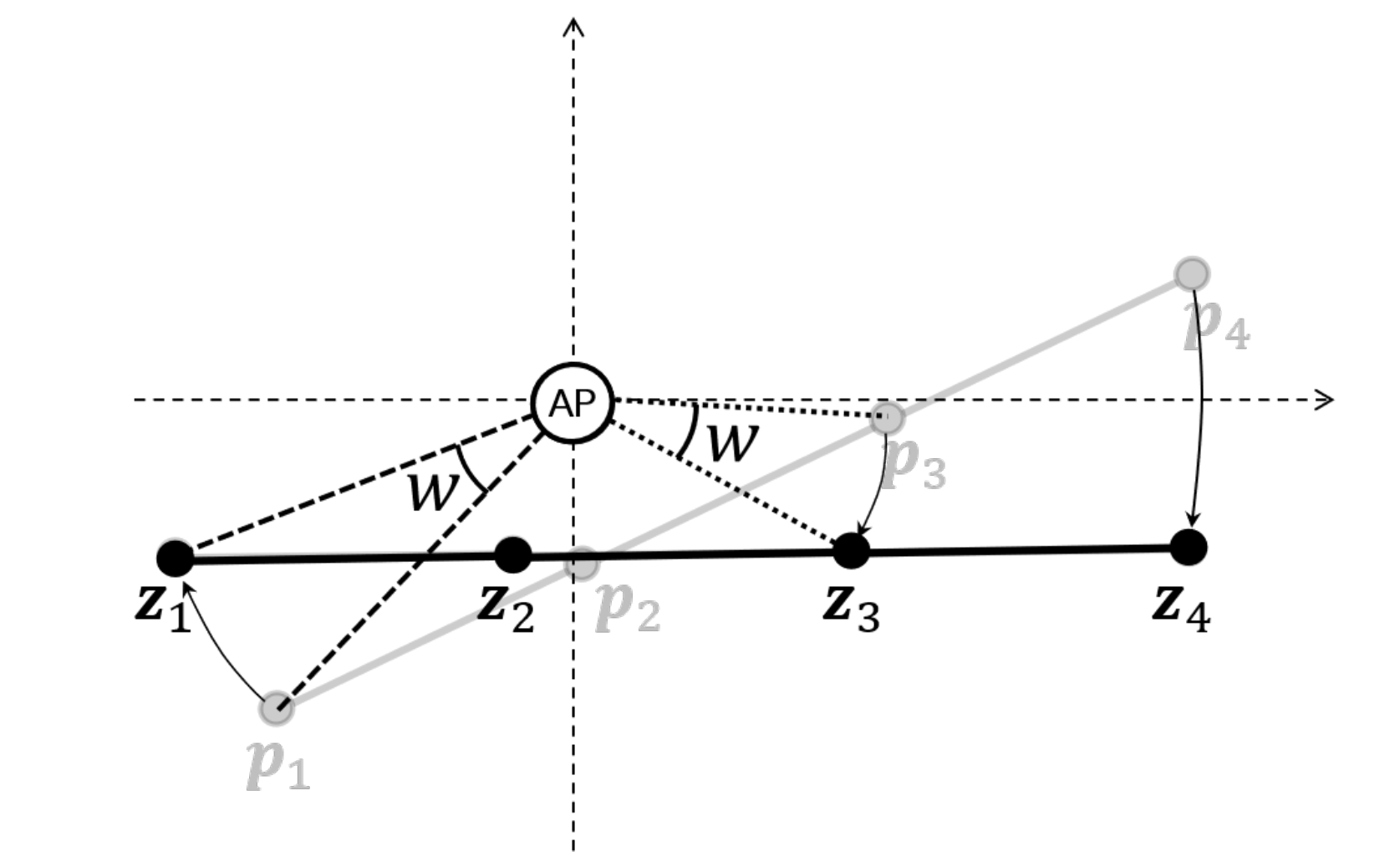}}
	\subfigure[Arbitrary mobility]{\includegraphics[width=7cm]{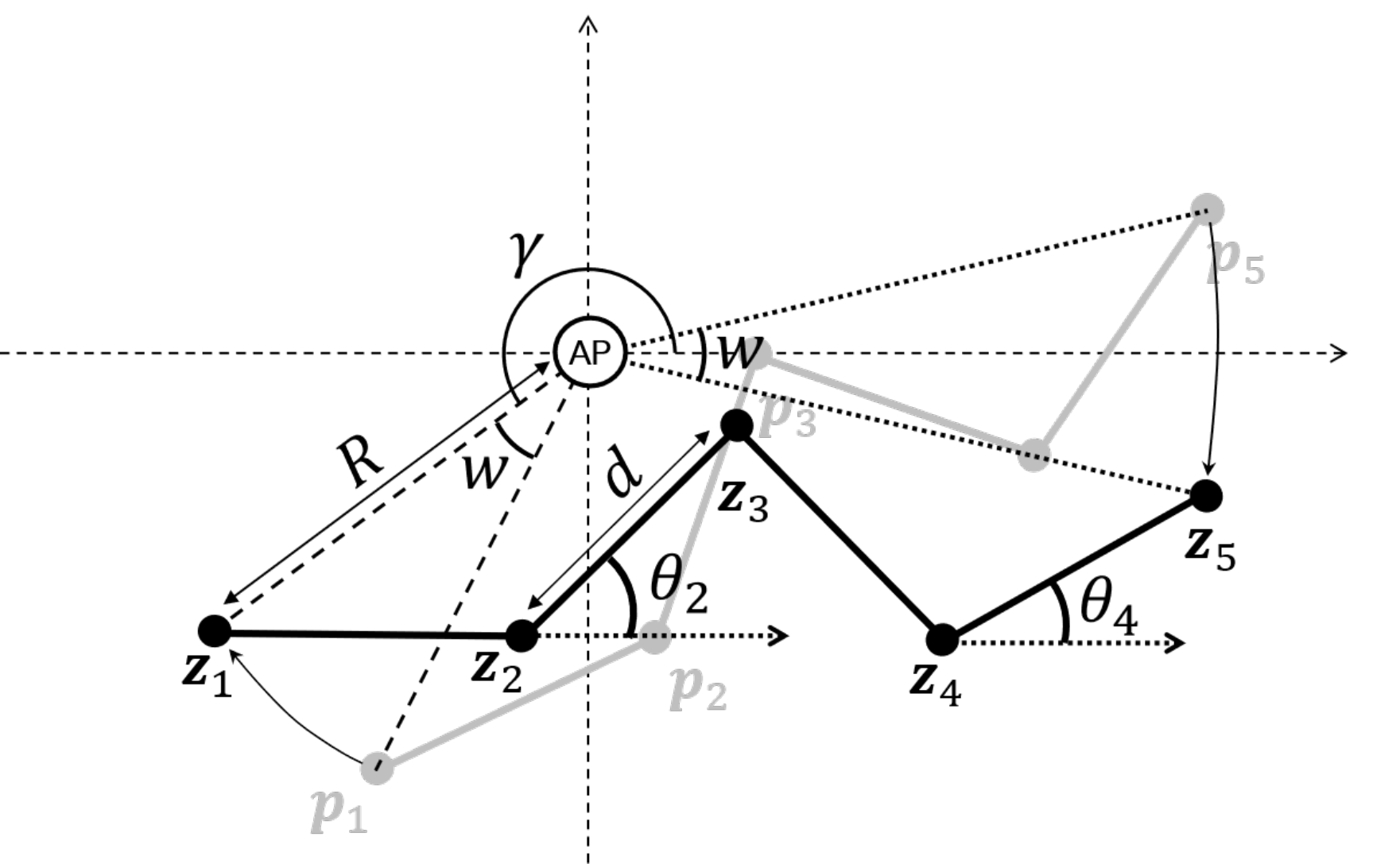}}
	\caption{Transformation to a local coordinate system whose origin is the AP's location and X-axis is aligned to the initial heading direction $\omega$.}
	\label{Fig:Two_case}
\end{figure}

\subsection{Transformation to a Local Coordinate System}

This section focuses on a pair of the user and a single WiFi AP.  For brevity, the location and the bias of the concerned AP are denoted by $\boldsymbol{p}_{\text{AP}}$ and $b$, respectively, by omitting the AP's index $m$. Then, we transform a global coordinate system into a local coordinate one by shifting the origin into $\boldsymbol{p}_{\text{AP}}$ and rotating the X-axis aligned with the heading direction $\omega$, as shown in Fig.~\ref{Fig:Two_case}.  Denote $\boldsymbol{z}_{n}=[q_n, u_n]^T$ the user's location redefined in the local coordinates. Then, 
\eqref{eq:temporal_relation} and \eqref{Eq:DistanceRelation} are rewritten as
\begin{align} \label{eq:temporal_relation_local}
\boldsymbol{z}_{n+1} =& \boldsymbol{z}_{n}+ d
\begin{bmatrix}
\cos(\theta_{n}), \sin(\theta_{n})
\end{bmatrix}^T\nonumber\\
=&\boldsymbol{z}_{1}+ d 
\begin{bmatrix}
\sum_{j=1}^{n} \cos(\theta_{j}), \sum_{j=1}^{n} \sin(\theta_{j})
\end{bmatrix}^T,
\end{align}
\begin{align} \label{Eq:DistanceRelation_local}
r_n=  \norm{\boldsymbol{z}_{n}} +b.
\end{align}
Plugging \eqref{eq:temporal_relation_local} into \eqref{Eq:DistanceRelation_local} gives
\begin{align}\label{Eq:ProblemFormulation1}
r_n= \norm[\bigg]{\boldsymbol{z}_{1}+  d 
\begin{bmatrix}
\sum_{j=1}^{n-1} \cos(\theta_{j}), \sum_{j=1}^{n-1} \sin(\theta_{j})
\end{bmatrix}^T} 
+b,
\end{align}
which is rewritten in terms of $q_1$ and $u_1$ as
\begin{align}
	\left(q_1+d c_n\right )^2 + \left(u_1+d s_n \right)^2 = (r_{n}-b)^2, \quad n\in\mathbb{N},
	\label{eq_nonlinearModel}
\end{align}
where $c_n=\sum_{j=1}^{n-1} \cos(\theta_{j})$ and $s_n=\sum_{j=1}^{n-1} \sin(\theta_{j})$. Depending on the user's mobility pattens, different algorithms are designed, introduced in the following subsections. 

\subsection{Case 1: Linear Mobility}\label{Subsec:LinearMobility}
Consider the case when the user is moving in a straight line without any direction change, which can be detectable by observing $\{\theta_n\}$, namely, $\theta_n=0$, $\forall n\in \mathbb{N}$ [see Fig. \ref{Fig:Two_case}(a)]. It makes $c_n=1$ and $s_n=0$ for all $n$, and \eqref{eq_nonlinearModel} can be reduced as
\begin{align}\label{Eq:LinearMobilityCase}
q_1^2+2 n q_1 d+n^2d^2+u_1^2=r_{n}^2-2 r_n b +b^2, \quad n\in\mathbb{N}.
\end{align}
Next, the nonlinear terms in \eqref{Eq:LinearMobilityCase} related to $q_1$ and $b$ are eliminated by choosing two \emph{reference steps} (RSs), denoted by the subset $\mathbb{S}=\{{a_1}, {a_2}\} \subset \mathbb{N}$. The algorithm covering the selection of RSs is explained in the sequel. Given $\mathbb{S}$, subtracting \eqref{Eq:LinearMobilityCase} of the reference steps $a_1$ or $a_2$ from that of the $n$-step gives
\begin{align}\label{Eq:LinearMobilityCase2}
2(n-a)q_1 d +(n^2-a^2)d^2=r_n^2-r_a^2-2b(r_n-r_a), \quad a\in\mathbb{S}.
\end{align}
The first term can be called out by manipulating the above two as
\begin{align} \label{Eq:LinearMobilityCase2}
d^2(a_1-a_2) + &b \cdot 2 \left(\frac{r_n-r_{a_1}}{n-a_1}-\frac{r_n-r_{a_2}}{n-a_2}\right)\nonumber\\
=&\frac{r_n^2-r_{a_1}^2}{n-a_1}-\frac{r_n^2-r_{a_2}^2}{n-a_2}, \quad n\in \mathbb{N}\cap\mathbb{S}^c, 
\end{align}
which is linear of $d^2$ and $b$. As a result, we formulate a system of linear equations with two unknowns $d^2$ and $b$ as
\begin{align}\label{Eq:Linearization_Linear}\tag{E1}
	\boldsymbol{A}(\mathbb{S})\boldsymbol{x}=\boldsymbol{b}(\mathbb{S}),
\end{align}
where $\boldsymbol{x}=[d^2,b]^T$. For matrix $\boldsymbol{A}(\mathbb{S})$ and vector $\boldsymbol{b}(\mathbb{S})$, we have\footnote{For ease of exposition, the matrix $\boldsymbol{A}$ and the vector $\boldsymbol{b}$ in \eqref{Eq:A_B_Linear} and \eqref{Eq:A_B_Arbitrary} are expressed assuming $\mathbb{S} \cap \{1, N\}=\emptyset$. }
\begin{align}\label{Eq:A_B_Linear}
&\boldsymbol{A}(\mathbb{S}) = 
\begin{bmatrix}
a_1-a_2 & 2 \left(\frac{r_1-r_{a_1}}{1-a_1}-\frac{r_1-r_{a_2}}{1-a_2}\right)\\ 
\vdots & \vdots\\ 
a_1-a_2 & 2 \left(\frac{r_N-r_{a_1}}{N-a_1}-\frac{r_N-r_{a_2}}{N-a_2}\right)
\end{bmatrix}\in \mathbb{R}^{(N-2)\times 2}, \nonumber\\
&\boldsymbol{b}(\mathbb{S})=
\begin{bmatrix}
\frac{r_1^2-r_{a_1}^2}{1-a_1}-\frac{r_1^2-r_{a_2}^2}{1-a_2}\\
\vdots \\
\frac{r_N^2-r_{a_1}^2}{N-a_1}-\frac{r_N^2-r_{a_2}^2}{N-a_2}
\end{bmatrix}\in \mathbb{R}^{(N-2)\times 1}.
\end{align}
Problem \ref{Eq:Linearization_Linear} comprises $(N-2)$ equations with two unknowns of $d$ and $b$. It is thus straightforward to provide the feasibility condition of \ref{Eq:Linearization_Linear} as follows.
\begin{proposition} [Feasible Condition: Linear Mobility]\label{Proposition1}
\emph{Problem \ref{Eq:Linearization_Linear} has a unique solution if the number of detected steps are at least $4$, namely, $N\geq 4$.}
\end{proposition}
With $N\geq 4$, \ref{Eq:Linearization_Linear} can be solved by 
\begin{align}\label{Solution_Linear}
\boldsymbol{x}^*(\mathbb{S})=&[(d^*(\mathbb{S}))^2, b^*(\mathbb{S})]^T\nonumber\\
=&[\boldsymbol{A}(\mathbb{S})^T\boldsymbol{A}(\mathbb{S})]^{-1}\boldsymbol{A}(\mathbb{S})^T\boldsymbol{b}(\mathbb{S}). 
\end{align}

In the presence of significant measurement noise, the matrix $\boldsymbol{A}(\mathbb{S})$ and the vector $\boldsymbol{b}(\mathbb{S})$ are corrupted, denoted by $\tilde{\boldsymbol{A}}(\mathbb{S})$ and $\tilde{\boldsymbol{b}}(\mathbb{S})$. Then, \ref{Eq:Linearization_Linear} is replaced as the following minimization problem:
\begin{align}\label{Eq:Problem_minimization_linear}
\boldsymbol{x}^*(\mathbb{S})&=\arg\min_{\boldsymbol{x}} \norm{\tilde{\boldsymbol{A}}(\mathbb{S})\boldsymbol{x}-\tilde{\boldsymbol{b}}(\mathbb{S})}\nonumber\\
&=[\tilde{\boldsymbol{A}}(\mathbb{S})^T\tilde{\boldsymbol{A}}(\mathbb{S})]^{-1}\tilde{\boldsymbol{A}}(\mathbb{S})^T\tilde{\boldsymbol{b}}(\mathbb{S}),
\end{align}
which has the same structure as \eqref{Solution_Linear}. 

Given the pair of $d^*(\mathbb{S})$ and $b^*(\mathbb{S})$, we derive the coordinates of $\boldsymbol{z}_1$, denoted by 
$\boldsymbol{z}_1^*(\mathbb{S})=[q_1^*(\mathbb{S}), u_1^*(\mathbb{S})]$.  
First, plugging $\{d^*(\mathbb{S}), b^*(\mathbb{S})\}$ of \eqref{Solution_Linear} or \eqref{Eq:Problem_minimization_linear} into \eqref{Eq:LinearMobilityCase2} leads to deriving $q_1^*(\mathbb{S})$ as  
\begin{align}\label{Eq:q1_Linear}
q_1^*(\mathbb{S}) =\frac{\sum_{a\in\mathbb{S}}\sum_{n\in\mathbb{N}, n\neq a} \frac{r_n^2-r_a^2-2b^*(\mathbb{S})(r_n-r_a)-(n^2-a^2)(d^*(\mathbb{S}))^2}{2(n-a) d^*(\mathbb{S})}}{2(N-1)}.
\end{align}
Next, there exist two possible solutions for $u_1^*(\mathbb{S})$ satisfying \eqref{Eq:LinearMobilityCase} as $u_1^*(\mathbb{S})$ and $-u_1^*(\mathbb{S})$, where
\begin{align}\label{Eq:u1_Linear}
u_1^*(\mathbb{S})&=\frac{1}{N}\sum_{n\in\mathbb{N}}\Big(r_{n}^2-2 r_n b^*(\mathbb{S}) +\big(b^*(\mathbb{S})\big)^2 \nonumber\\
&\quad-\big(q_1(\mathbb{S})\big)^2 +2 n q_1(\mathbb{S}) d^*(\mathbb{S})+n^2\big(d^*(\mathbb{S})\big)^2\Big)^{\frac{1}{2}}. 
\end{align}
It is challenging to distinguish which one is correct under the current setting of a single WiFi AP. On the other hand, it is possible to discriminate the correct one if multiple APs are given, explained in the next section.

\subsection{Case 2: Arbitrary Mobility} \label{Subsec:ArbitraryMobility}
Consider the case when the user is randomly moving, namely, $\theta_n\neq0$ for some $n\in\mathbb{N}$ [see Fig. \ref{Fig:Two_case}(b)].
Letting $R=\sqrt{q_1^2+u_1^2}$ and $\gamma=\tan^{-1}\left(\frac{u_1}{q_1}\right)$ makes \eqref{eq_nonlinearModel} as
\begin{align}
&q_1^2+u_1^2+2d(q_1c_n+u_1s_n)+d^2(c_n^2+s_n^2)\nonumber\\
	\overset{(a)}{=}&R^2+2d R \left(\sum_{i=1}^{n-1}\cos(\theta_{i}-\gamma)\right)+d^2(c_n^2+s_n^2)\nonumber\\
	=&r_{n}^2-2 r_n b+b^2, 
	\label{eq_nonlinearModel-polar1}
\end{align}
where $(a)$ follows from $q_1 \cos(\theta_n)+u_1 \sin(\theta_n)=R\cos(\theta_n-\gamma)$. 
Using RSs $\mathbb{S}=\{a_1, a_2\}$ as in the previous case of linear mobility, we have  
\begin{align}\label{Eq:Intermediate}
	&2d R f_{n,a}(\gamma) + d^2 \eta_{n,a} \nonumber\\
	=& r_n^2-r_a^2-2b(r_n-r_a), \quad a\in \mathbb{S}, \quad n\in \mathbb{N}\cap\mathbb{S}^c,
	\end{align}
where $f_{n,a}(\gamma)=\sum_{i=1}^{n-1}\cos(\theta_{i}-\gamma)-\sum_{i=1}^{a-1}\cos(\theta_{a}-\gamma)$ and $\eta_{n,a}= (c_n^2+s_n^2)-(c_a^2+s_a^2)$. The nonlinear term $2d R f_{n,a}(\gamma)$ is canceled out by manipulating the above two as
\begin{align}\label{Eq:Intermediate2}
&d^2\underbrace{\left[f_{n,a_2}(\gamma)\eta_{n,a_1}-f_{n,a_1}(\gamma)\eta_{n,a_2}\right]}_{=\alpha_n(\mathbb{S}, \gamma)}\nonumber\\
&+b \cdot \underbrace{2\left[f_{n,a_2}(\gamma)(r_n-r_{a_1})-f_{n,a_1}(\gamma)(r_n-r_{a_2})\right]}_{=\beta_n(\mathbb{S}, \gamma)}\nonumber\\
=&\underbrace{f_{n,a_2}(\gamma)(r_n^2-r_{a_1}^2)-f_{n,a_1}(\gamma)(r_n^2-r_{a_2}^2)}_{=\zeta_n(\mathbb{S}, \gamma)}, \quad  n \in \mathbb{N}\cap\mathbb{S}^c. 
\end{align}
We formulate another system of linear equations with two unknowns $\boldsymbol{x}=[d^2,b]^T$ as
\begin{align}\label{Eq:Linearization_Arbitrary}\tag{E2}
	\boldsymbol{A}(\mathbb{S}, \gamma)\boldsymbol{x}=\boldsymbol{b}(\mathbb{S}, \gamma),
\end{align}
where 
\begin{align}\label{Eq:A_B_Arbitrary}
&\boldsymbol{A}(\mathbb{S}, \gamma) = 
\begin{bmatrix}
\alpha_1(\mathbb{S}, \gamma) & \beta_1(\mathbb{S}, \gamma)\\ 
\vdots & \vdots\\ 
\alpha_N(\mathbb{S}, \gamma) & \beta_n(\mathbb{S}, \gamma)
\end{bmatrix}\in \mathbb{R}^{(N-2)\times 2}, \nonumber\\
&\boldsymbol{b}(\mathbb{S}, \gamma)=
\begin{bmatrix}
\zeta_1(\mathbb{S}, \gamma)\\
\vdots \\
\zeta_n(\mathbb{S}, \gamma)
\end{bmatrix}\in \mathbb{R}^{(N-2)\times 1}.
\end{align}
Given $\mathbb{S}$, \ref{Eq:Linearization_Arbitrary} has (N-2) equations with three unknowns of $d$, $b$, and $\gamma$, providing the following feasible condition. 

\begin{proposition} [Feasible Condition: Arbitrary Mobility]\label{Proposition2}\emph{Unless $f_{a_1, a_2}(\gamma)=0$, Problem \ref{Eq:Linearization_Arbitrary} has a unique solution if the number of detected steps are at least $5$, namely, $N\geq 5$.}
\end{proposition}
\begin{proof}
See Appendix \ref{Appen:Propostion2}. 
\end{proof}

\begin{remark}[Underdetermined System]\emph{Define $\hat{\gamma}$ the angle satisfying $f_{a_1, a_2}(\hat{\gamma})=0$. The rank of $\boldsymbol{A}(\mathbb{S}, \hat{\gamma})$ is $1$, making \ref{Eq:Linearization_Arbitrary} an underdetermined system that has infinite number of solutions. Using the condition {$f_{a_1,a_2}(\gamma)=0$}, it is straightforward to identify whether the concerned $\gamma$ is $\hat{\gamma}$.}
\end{remark}

With  $N\geq 5$ and $f_{a_1, a_2}(\gamma)\neq0$, the solution for \ref{Eq:Linearization_Arbitrary} has a similar form to the linear mobility counterpart as
\begin{align}\label{Eq:Solution_Structure_Arbitrary}
\boldsymbol{x}(\mathbb{S},\gamma)&=[(d(\mathbb{S},\gamma))^2, b(\mathbb{S},\gamma)]^T\nonumber\\
&=[\boldsymbol{A}(\mathbb{S},\gamma)^T\boldsymbol{A}(\mathbb{S},\gamma)]^{-1}\boldsymbol{A}(\mathbb{S},\gamma)^T\boldsymbol{b}(\mathbb{S}, \gamma),
\end{align}
which is valid only when $\gamma$ is correctly picked. Otherwise, $\boldsymbol{x}(\mathbb{S},\gamma)$ of \eqref{Eq:Solution_Structure_Arbitrary} does not satisfy \ref{Eq:Linearization_Arbitrary}, i.e., $\boldsymbol{A}(\mathbb{S}, \gamma)\boldsymbol{x}(\mathbb{S},\gamma)\neq \boldsymbol{b}(\mathbb{S}, \gamma)$. Prompted by the fact, a correct $\gamma^*(\mathbb{S})$ can be easily found by a simple 1D search over $[0,\pi)$ satisfying the following criteria:
\begin{align}\label{discriminant_noise_free}
\gamma^*(\mathbb{S})=\{\gamma | e_1(\mathbb{S}, \gamma)=0, \quad \text{$\gamma\in[0,\pi)$}\}, 
 \end{align}
 where 
 \begin{align}\label{Eq:Cost_Function_1}
 e_1(\mathbb{S}, \gamma)=\norm{\boldsymbol{A}(\mathbb{S}, \gamma)\boldsymbol{x}(\mathbb{S},\gamma)- \boldsymbol{b}(\mathbb{S}, \gamma)}.
 \end{align}
 
 \begin{remark}[Ambiguity of $\gamma$]\emph{Noting the period of $\tan^{-1}(x)$ being $\pi$, two possible solutions of $\gamma^*(\mathbb{S})$ exist in $[0, \pi)$ and $[\pi, 2\pi)$, say $\gamma_1^*(\mathbb{S})$ and $\gamma_2^*(\mathbb{S})$,
 where $\gamma_1^*(\mathbb{S})+\pi=\gamma_2^*(\mathbb{S})$. Despite the ambiguity, the resultant solutions of $d$ and $b$ are not changed regardless of $\gamma_1^*(\mathbb{S})$ or $\gamma_2^*(\mathbb{S})$, namely,
$d(\mathbb{S},\gamma_1^*(\mathbb{S}))=d(\mathbb{S},\gamma_2^*(\mathbb{S}))$ and  $b(\mathbb{S},\gamma_1^*(\mathbb{S}))=b(\mathbb{S},\gamma_2^*(\mathbb{S}))$.
As a result, we focus on finding $\gamma_1^*(\mathbb{S})$ in the range of $[0, \pi)$, and it is considered as $\gamma^*(\mathbb{S})$. }
\end{remark}

 When the measurement are corrupted, the noisy versions of the matrix $\boldsymbol{A}(\mathbb{S}, \gamma)$ and the vector $\boldsymbol{b}(\mathbb{S}, \gamma)$ are given, denoted by $\tilde{\boldsymbol{A}}(\mathbb{S}, \gamma)$ and $\tilde{\boldsymbol{b}}(\mathbb{S}, \gamma)$,  making it difficult to use the discriminant of \eqref{discriminant_noise_free} directly. Instead, we develop the following two-stage approach. 
\subsubsection{Finding $\boldsymbol{x}$} Given $\gamma$, we formulate the following minimization problem as
\begin{align}\label{Eq:Problem_minimization_arbitrary}
\boldsymbol{x}(\mathbb{S}, \gamma)&=\arg\min_{\boldsymbol{x}} \norm{\tilde{\boldsymbol{A}}(\mathbb{S}, \gamma)\boldsymbol{x}-\tilde{\boldsymbol{b}}(\mathbb{S}, \gamma)}\nonumber\\
&=[\tilde{\boldsymbol{A}}(\mathbb{S},\gamma)^T\tilde{\boldsymbol{A}}(\mathbb{S},\gamma)]^{-1}\tilde{\boldsymbol{A}}(\mathbb{S},\gamma)^T\tilde{\boldsymbol{b}}(\mathbb{S}, \gamma),
\end{align}
which follows an equivalent structure of \eqref{Eq:Solution_Structure_Arbitrary}. 
\subsubsection{Finding $\gamma$} We use a 1D search based on the following criteria:  
{
\begin{align}\label{Eq:FindingGamma_arbitraty}
\gamma^*(\mathbb{S})=\arg\min_{\gamma\in\mathbb{P}(\mathbb{S})} 
\left[w_1 \cdot e_1(\mathbb{S}, \gamma)+w_2 \cdot e_2(\mathbb{S}, \gamma)\right],
\end{align}
where $\mathbb{P}(\mathbb{S})$ is a range of feasible $\gamma$ defined as all calibrated distances $\{r_n-b(\mathbb{S}, \gamma)\}$ and estimated step length $d(\mathbb{S}, \gamma)$ being positive, namely,
\begin{align} \label{Eq:PossibleSet}
	\mathbb{P}=\left\{\gamma\left\vert \min_{n\in\mathbb{N}}[r_n-b(\mathbb{S}, \gamma)]>0,\right. \quad d(\mathbb{S}, \gamma)>0, \quad \gamma\in[0,\pi)\right\}.
\end{align}
Two error functions are considered whose weighed factors $\{w_1,w_2\}$ satisfy $w_1+w_2=1$. 
The first function $e_1(\mathbb{S}, \gamma)$ is specified in \eqref{Eq:Cost_Function_1}. For the second one $e_2(\mathbb{S}, \gamma)$, denote the matrix $\boldsymbol{R}(\mathbb{S}, \gamma)=[R_{1,a_1}(\mathbb{S}, \gamma), \cdots, R_{N,a_1}(\mathbb{S}, \gamma);R_{1,a_2}(\mathbb{S}, \gamma), \cdots, R_{N,a_2}(\mathbb{S}, \gamma)]\in\mathbb{R}^{2\times(N-1)}$, where $R_{n,a}(\mathbb{S}, \gamma)$ is the estimated $R$ obtained by plugging $d(\mathbb{S}, \gamma)$ and $b(\mathbb{S}, \gamma)$ into the $n$-th equation of \eqref{Eq:Intermediate} as
\begin{align}
&R_{n,a}(\mathbb{S}, \gamma)  = \frac{r_n^2-r_a^2-2 (r_n-r_a)\cdot b(\mathbb{S}, \gamma)-(d(\mathbb{S}, \gamma))^2 \eta_{n,a}}{2f_{n,a}(\gamma)\cdot d(\mathbb{S}, \gamma)} \nonumber\\ 
&\quad , \quad a\in\mathbb{S}.  
\end{align}
Given $\boldsymbol{R}(\mathbb{S}, \gamma)$, the error function $e_2(\mathbb{S}, \gamma)$ is defined as
\begin{align}\label{Eq:Cost_Function_2}
e_2(\mathbb{S}, \gamma)=\mathsf{std}(\boldsymbol{R}(\mathbb{S}, \gamma)),
\end{align}
where $\mathsf{std}(\cdot)$ represents the standard deviation of all elements therein.
}
\begin{remark}[Error Functions] \label{Remark:EffectOfErrorFunction2}
\emph{Given $\mathbb{S}$, the first error function $e_1(\mathbb{S}, \gamma)$ represents the error of  \ref{Eq:Linearization_Arbitrary} itself by focusing its explicit solution $\boldsymbol{x}=[d^2, b]^T$.  
On the other hand, the second error function $e_2(\mathbb{S}, \gamma)$ captures the error on a latent variable $R$ that is eliminated in \ref{Eq:Linearization_Arbitrary} due to the linearization process of \eqref{Eq:Intermediate} and \eqref{Eq:Intermediate2}. Selecting the weight factors $\{w_1, w_2\}$ is discussed in Section~\ref{Sec:Experiement}. 
}
\end{remark}

In both cases with and without measurement noises, a pair of the optimal solution $d^*(\mathbb{S})$ and $b^*(\mathbb{S})$ can be obtained from the solution $\boldsymbol{x}^*(\mathbb{S})=\boldsymbol{x}(\mathbb{S}, \gamma^*(\mathbb{S}))$. Given $\{d^*(\mathbb{S}), b^*(\mathbb{S})\}$, we derive the coordinates of $\boldsymbol{z_1}$, say $\boldsymbol{z}_1^*(\mathbb{S})=[q_1^*(\mathbb{S}), u_1^*(\mathbb{S})]^T$ as follows. First, linear equations with $q_1$ and $u_1$ are derived from \eqref{Eq:Intermediate} as
\begin{align}
&q_1(c_n-c_a)+u_1(s_n-s_a) \nonumber\\
=&\underbrace{ \frac{r_n^2-r_a^2-2b^*(\mathbb{S})\cdot (r_n-r_a)-(d^*(\mathbb{S}))^2 \eta_{n,a}}{2d^*(\mathbb{S})}}_{=g_{n,a}(\mathbb{S})}, \quad a\in\mathbb{S}, n\neq a,
\end{align}
leading to formulating a system of linear equations as
\begin{align}\label{Eq:ProblemFormulation3}\tag{E3}
{\boldsymbol{H}(\mathbb{S})}\boldsymbol{z}_1=\boldsymbol{g}(\mathbb{S}),
\end{align}
where ${\boldsymbol{H}(\mathbb{S})}=\begin{bmatrix}
\boldsymbol{H}_{a_1}; \boldsymbol{H}_{a_2}
\end{bmatrix}\in\mathbb{R}^{2(N-1)\times 2}$ and $\boldsymbol{g}(\mathbb{S})=\begin{bmatrix}
\boldsymbol{g}_{a_1}(\mathbb{S}); \boldsymbol{g}_{a_2}(\mathbb{S})
\end{bmatrix}\in\mathbb{R}^{2(N-1)\times 1}$ with
\begin{align}\label{Eq:H_v_Arbitrary}
&\boldsymbol{H}_a = 
\begin{bmatrix}
c_1-c_a & s_1-s_a\\ 
\vdots & \vdots\\ 
c_N-c_a & s_N-s_a
\end{bmatrix}\in \mathbb{R}^{(N-1)\times 2}, \nonumber\\
&\boldsymbol{g}_a(\mathbb{S})=
\begin{bmatrix}
g_{1,a}(\mathbb{S})\\
\vdots \\
g_{N,a}(\mathbb{S})
\end{bmatrix}\in \mathbb{R}^{(N-1)\times 1},\quad a\in\mathbb{S}. 
\end{align}
Given the feasible condition of \ref{Eq:Linearization_Arbitrary} stated in Proposition \ref{Proposition2},  \ref{Eq:ProblemFormulation3} has a unique solution obtained by a single matrix inversion as
\begin{align}\label{Eq:FindingLocalCoordinates}
\boldsymbol{z}_1^*(\mathbb{S})=&[q_1^*(\mathbb{S}), u_1^*(\mathbb{S})]^T\nonumber\\
=&\arg\min_{\boldsymbol{z}_1}\norm{{\boldsymbol{H}(\mathbb{S})}\boldsymbol{z}_1- \boldsymbol{g}(\mathbb{S})} \nonumber\\
=&{\left[\boldsymbol{H}(\mathbb{S})^T\boldsymbol{H}(\mathbb{S})\right]^{-1}\boldsymbol{H}(\mathbb{S})^T\boldsymbol{g}(\mathbb{S}).}
\end{align}

\subsection{Using Multiple Combinations of Reference Steps}\label{Subsec:RS_Selection} 

This subsection deals with the remaining issue of selecting RSs $\mathbb{S}$, helping mitigate the positioning error due to significant measurement noises. To this end, multiple combinations of RSs are utilized to achieve more accurate positioning than a single RS-based scheme, based on a common statistical belief that more observations make the estimate less deviated from a ground-truth. The detailed procedure is explained as follows. 
\subsubsection{Selecting Candidate RSs}
First, several steps are picked as RS's candidates, denote by $\mathbb{C}$, based on the initial propagation distance {estimates} $\{r_n\}$ defined in \eqref{Eq:DistanceRelation}. In general, smaller $r_n$ means that AP $n$ is located in proximity whose RSS is likely to be high. It is thus reasonable to consider the {resultant ranging result} is relatively accurate. Motivated by this intuition, the set $\mathbb{C}$ contains a step's index, say $n$, if $r_n$ is in the top $C$ smallest, namely, 
\begin{align}
\mathbb{C}=\{n \in \mathbb{N} | r_n\leq r_k, n\in \mathbb{C}, k\in\mathbb{C}^c, |\mathbb{C}|=C\},
\end{align} 
where $C$ is the cardinality constraint of $\mathbb{C}$, whose effect is verified by field experiments in Section~\ref{Sec:Experiement}. 

{
\subsubsection{Estimating the Bias and Step Length} Two indices of $\mathbb{C}$ are picked as $\mathbb{S}$. It is possible to make up to $L={C \choose 2}$ combinations of $\mathbb{S}$. Denote $\mathbb{S}_\ell$ the $\ell$-th set of RSs, $\ell\in\{1, \cdots, L\}$. Given $\mathbb{S}_\ell$, compute $d^*(\mathbb{S}_\ell)$ and $b^*(\mathbb{S}_\ell)$ by following the procedures in Sec. \ref{Subsec:LinearMobility} and Sec. \ref{Subsec:ArbitraryMobility}, depending on the cases of linear and arbitrary mobilities, respectively. Given all individual estimates $\{d^*(\mathbb{S}_\ell)\}$ and $\{b^*(\mathbb{S}_\ell)\}$ representative estimates, denoted by $d^*$ and $b^*$ respectively, are computed using their medians, namely, 
\begin{align}\label{Eq:Optimal_Solution_Median}
	d^*=\mathsf{median}({d}^*(\mathbb{S}_\ell)),\quad b^*=\mathsf{median}({b}^*(\mathbb{S}_\ell)).
\end{align}
\subsubsection{Estimating the relative coordinates of $\boldsymbol{z}_1$} Given $d^*$ and $b^*$, 
compute $\{\boldsymbol{z}_1^*({\mathbb{S}_\ell})\}$ for all possible sets of RSs using \eqref{Eq:q1_Linear} and \eqref{Eq:u1_Linear} for the case of linear mobility, or \eqref{Eq:FindingLocalCoordinates} for the case of arbitrary mobility. Given all individual estimates $\{\boldsymbol{z}_1^*(\mathbb{S}_\ell)\}$, representative estimate, denoted by $\boldsymbol{z}_1^*$ is computed using their medians, namely,
\begin{align}\label{Eq:Optimal_Solution_Median2}
\boldsymbol{z}_1^*=\mathsf{median}(\boldsymbol{z}_1^*(\mathbb{S}_\ell)). 
\end{align}
}
\begin{remark}[Mean vs. Median]\emph{While a mean-based estimation has been widely used as a de facto standard approach, it is prone to a few outliers severely deviated from a ground-truth value. On the other hand, a median-based estimation can ignore these outliers. Thus, it is more suitable to design a positioning algorithm based on the median, which is simple yet robust from measurement noises, such as \cite{Casas2006} and \cite{Qiao2014}.  }
\end{remark}

\section{Positioning via Trajectory Alignment}\label{Sec:TrajectoryAlignment}

In this section, we aim at positioning the user's locations by aligning multiple trajectories based on the measurements of different APs, called \emph{trajectory alignment} (TA). First, the user's relative trajectory defined on the local coordinate system of each AP is derived based on the estimations in the preceding section. Next, a basic principle of TA is mathematically explained assuming the case without measurement noise. Last, a practical algorithm is designed able to work in the case with measurement noise. 

\subsection{Relative Trajectory Derivation}  

This section derives the sequence of the user's locations, denoted by $\mathcal{Z}=\{\boldsymbol{z}^*_n\}=\{[q_n^*, u_n^*]\}$, corresponding to the user's \emph{relative trajectory} defined on the local coordinate system. 
From the preceding estimations of the initial location $\boldsymbol{z}_1^*=[q_1^*, u_1^*]$ in \eqref{Eq:Optimal_Solution_Median2}, it is possible to derive the following locations using \eqref{eq:temporal_relation_local} and the step length estimation $d^*$. Depending on the case of linear or arbitrary mobilities, we have different results explained below.  

\subsubsection{Linear mobility} Recalling that there exist two candidates of $u_1^*$ [see \eqref{Eq:u1_Linear}], two possible local trajectories are thus made, say $\mathcal{Z}_{+}$ and $\mathcal{Z}_-$, given as
\begin{align} \label{Eq:LocalTrajectoiesLinear}
\mathcal{Z}_{+}&=\left\{\l[q_n^*, u_n^*\r] \vert q_n^*=q_1^*+(n-1)d,\quad u_n^*=+u_1^*, \quad \forall n\in\mathbb{N} \right\},\nonumber\\
\mathcal{Z}_{-}&=\left\{\l[q_n^*, u_n^*\r] \vert q_n^*=q_1^*+(n-1)d,\quad u_n^*=-u_1^*, \quad \forall n\in\mathbb{N} \right\},
\end{align}
where $q_1^*$ and $u_1^*$ are specified in \eqref{Eq:q1_Linear} and \eqref{Eq:u1_Linear}, respectively. 
Either $\mathcal{Z}_{+}$ or $\mathcal{Z}_{-}$ is the real trajectory $\mathcal{Z}$, differentiated by the positioning algorithm introduced in the sequel.
\subsubsection{Arbitrary mobility} Contrary to the linear mobility counterpart, no ambiguity of the initial location exists. The resultant local trajectory $\mathcal{Z}$ is given as
\begin{align}
\mathcal{Z}=\left\{\l[q_n^*, u_n^*\r] \left \vert q_n^*=q_1^*+d c_n, \quad u_n^*=u_1^*+d s_n, \right. \quad \forall n\in\mathbb{N}\right\}, 
\end{align}
where the coefficient $c_n$ and $s_n$ are specified in \eqref{eq_nonlinearModel}.

\subsection{Trajectory Alignment}

\begin{figure}[t]
	\centering
	\subfigure[Local Trajectories]{\includegraphics[width=7cm]{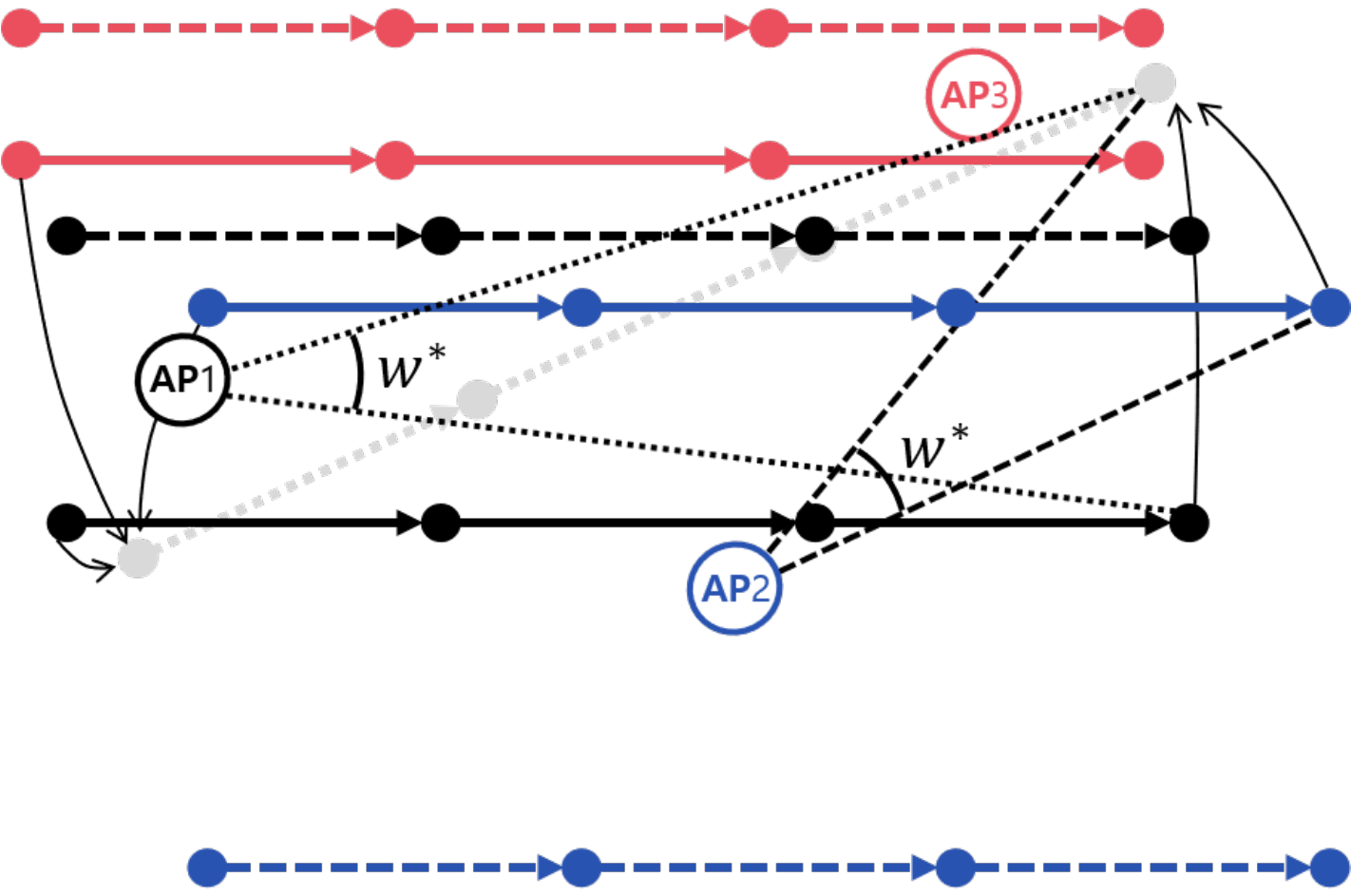}} 
	\subfigure[Global Trajectories]{\includegraphics[width=7cm]{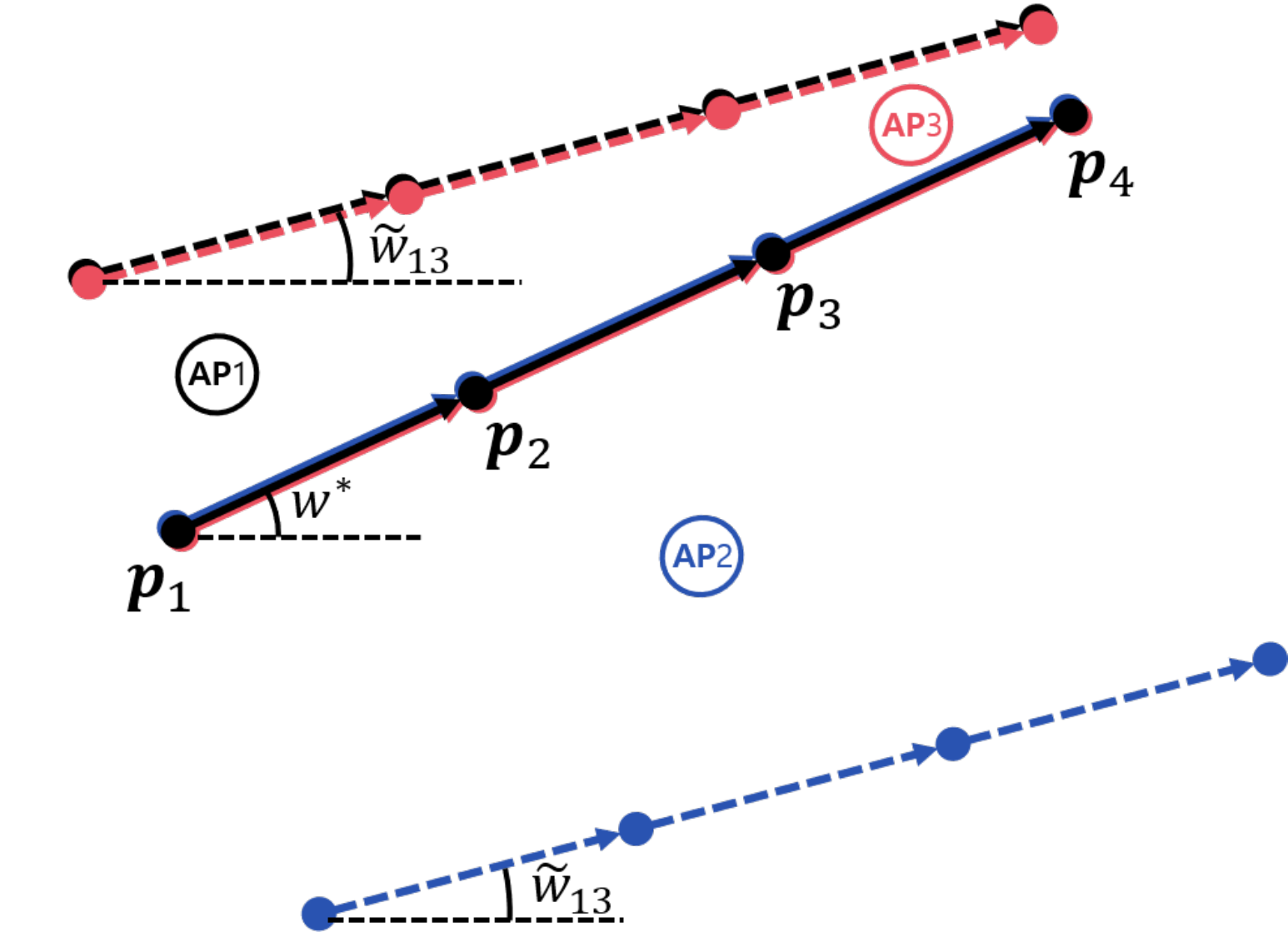}}
	\caption{The graphical example of TA in the case of linear mobility. 
	Solid and dotted lines represent two possible trajectories, which are symmetric to each other.   
	The real trajectories are merged into one in global coordinate, but the others are not. }
	\label{Fig:TA_Linear}
\end{figure}

\begin{figure}[t]
	\centering
	\subfigure[Local Trajectories]{\includegraphics[width=7cm]{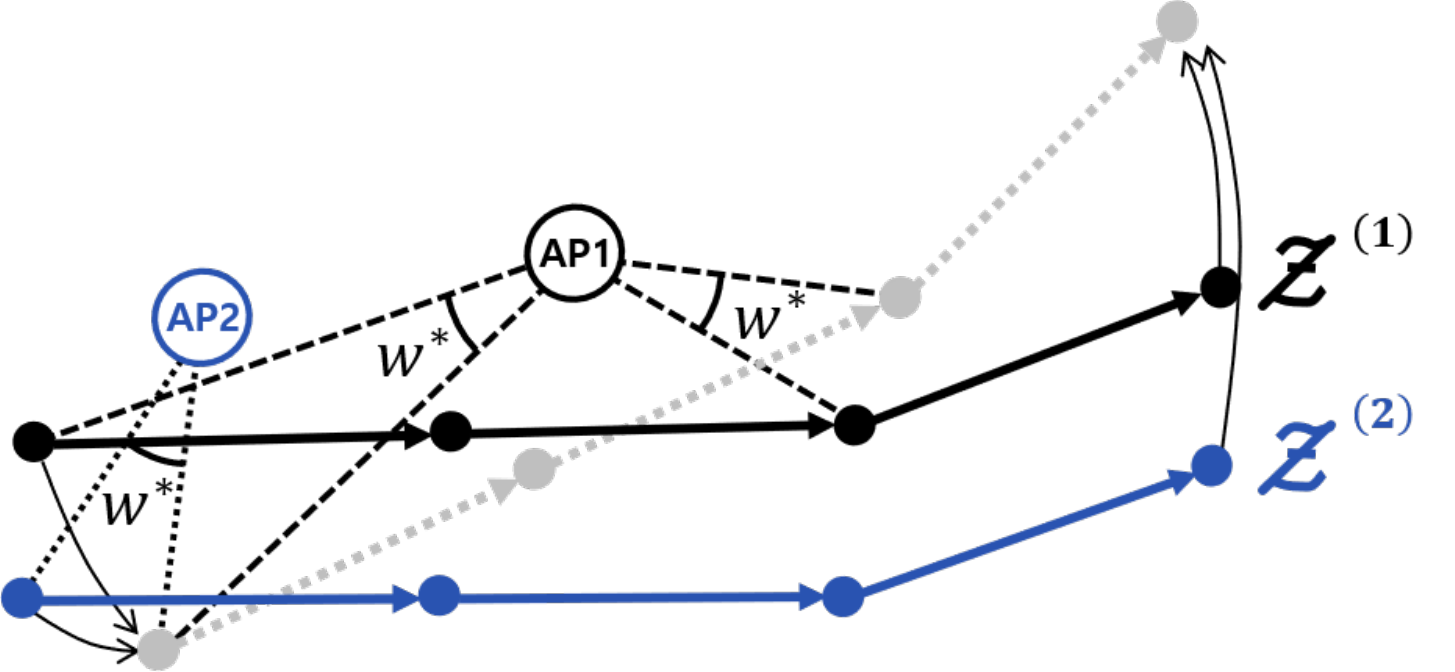}} 
	\subfigure[Global Trajectories]{\includegraphics[width=7cm]{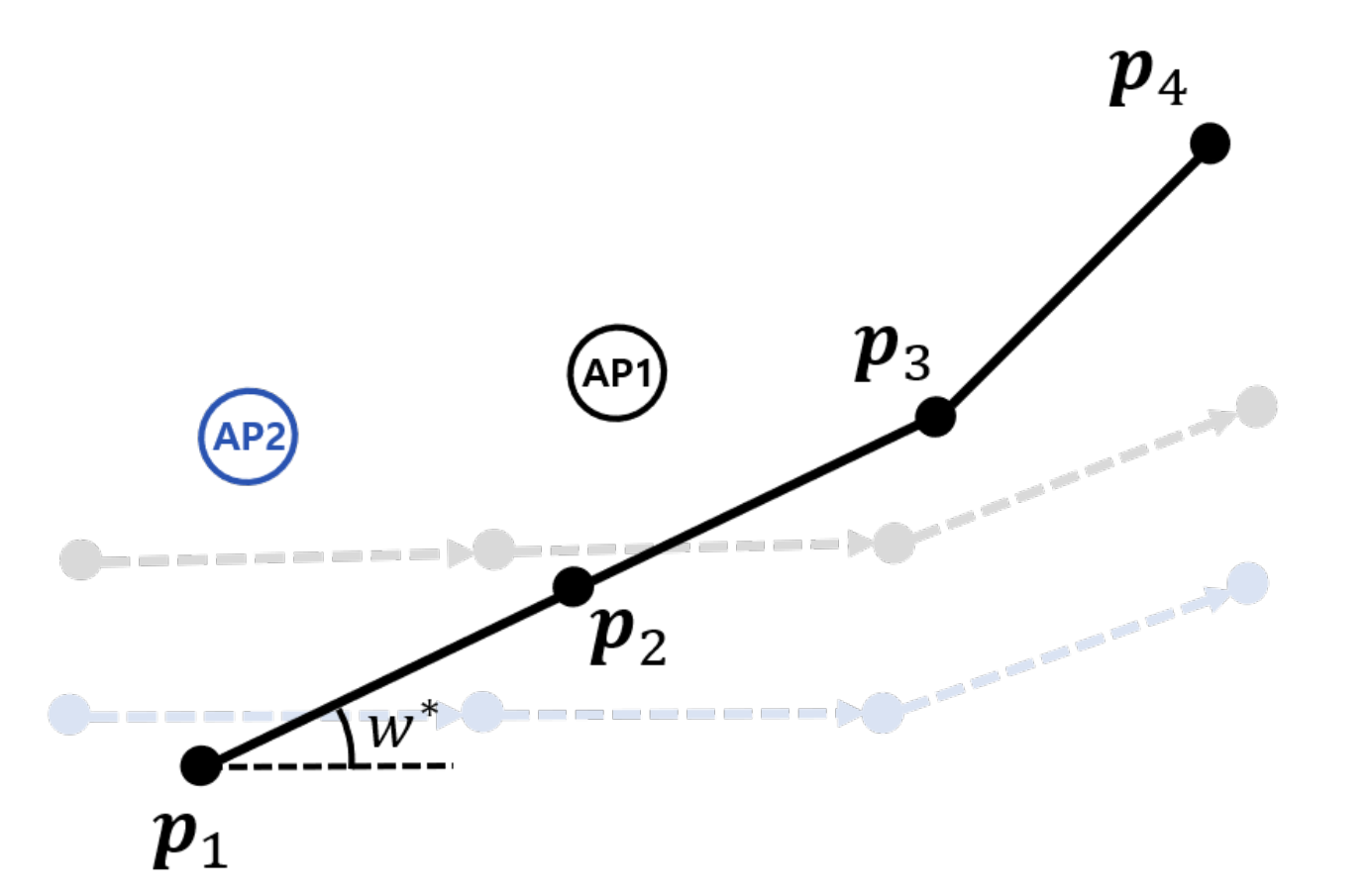}}
	\caption{The graphical example of TA in the case of arbitrary mobility, showing that trajectories in local coordinates represented by multiple curves have a unique angle $\omega^*$, making them aligned perfectly. }
	\label{Fig:TA_Arbitary}
\end{figure}

This subsection introduces the concept of TA and explains its feasible conditions.  
Consider the user's relative trajectory estimated by the RTT measurements from AP $m$, say $\mathcal{Z}^{(m)}=\{\boldsymbol{z}^{(m)*}_n\}$. It is converted into a global coordinate system when the initial direction $\omega$ is given, namely, 
\begin{align} \label{Eq:GlobalPostion}
\boldsymbol{p}_n^{(m)}(\omega)=
\boldsymbol{p}_{\text{AP}}^{(m)}+
\begin{bmatrix}
\cos(\omega) & -\sin(\omega)\\
\sin(\omega) & \cos(\omega)
\end{bmatrix}\boldsymbol{z}^{(m)*}_n, 
\end{align}
where $\boldsymbol{p}_{\text{AP}}^{(m)}$ is AP $m$'s location assumed to be known in advance. 
Noting that the user's location is unique regardless of which AP is used for positioning, the following condition should be met if $\omega$ is correctly selected, denoted by $\omega^*$:  
\begin{align}\label{Eq:EQ_Condition}
\boldsymbol{p}_n=\boldsymbol{p}_n^{(1)}(\omega^*)=\boldsymbol{p}_n^{(2)}(\omega^*)=\cdots=\boldsymbol{p}_n^{(M)}(\omega^*), \quad \forall n\in\mathbb{N}, 
\end{align}
which is said that all trajectories are perfectly aligned. The aligned trajectory after TA is equivalent to the trajectory defined on the global coordinate system, denoted by $\mathcal{P}=\{\boldsymbol{p}_n\}$, if it exists uniquely. The following proposition gives different feasible conditions of TA for linear and arbitrary mobilities. 
\begin{proposition}[Feasible Condition of Trajectory Alignment]\label{Proposition3}\emph{There exists a unique $\omega^*$ satisfying the condition of \eqref{Eq:EQ_Condition} if the number of APs $M$ not on a straight line is at least $3$ for linear mobility or the number of APs $M$ is at least $2$ for arbitrary mobility. }
\end{proposition}
\begin{proof}
See Appendix \ref{Appen:Propostion3}. 
\end{proof}

Fig. \ref{Fig:TA_Linear} and \ref{Fig:TA_Arbitary} respectively represent the graphical examples of TA for linear and arbitrary mobilities, showing that one more AP is required to identify whether $\mathcal{Z}^{(m)}$ is $\mathcal{Z}_{+}^{(m)}$ or $\mathcal{Z}_{-}^{(m)}$.  The comparison between the two from the perspective of entire procedure is discussed in the following remarks.  
\begin{remark}[Linear vs. Arbitrary Mobilities]\emph{{The different feasible conditions for linear mobility and arbitrary mobility stem from the difference of location dimension embedded therein. Linear mobility is interpreted as 1D location information in a 2D space, bringing about the ambiguity issue illustrated in Fig. \ref{Fig:TA_Linear}. On the other hand, arbitrary mobility provides 2D location information, facilitating TA without ambiguity. This difference yields the positioning accuracy gap between the two, verified in Sec. \ref{subsec:PositioningAccuracy}.}}
\end{remark}

\subsection{Algorithm Design} \label{Subsec:Algorithm_Design}

In an ideal case without a measurement noise, it is possible to find $\omega^*$ making all trajectories aligned perfectly, equivalent to satisfying the condition \eqref{Eq:EQ_Condition}. In contrast, it may be challenging to do in practical cases with measurement noises, since several APs rather deteriorates the positioning accuracy if their estimation errors of bias and step length are severe. It is overcome by excluding these APs in advance and minimizing a new error function defined for TA. The detailed algorithm is explained below.

\subsubsection{Feasible AP Selection} First, we aim at excluding APs unlikely to contribute {to} accurate positioning\footnote{{Our feasible AP selection is based on the assumption that RTT from all APs are measurable during walking. In the coexistence of dynamic and hotspot APs, it is required to filter them out for a reliable positioning result, like the fingerprint filtering method in \cite{Bisio2019}.}}. To this end, we define the set of feasible APs $\mathbb{F}$, whose element's bias and step length estimations, say $b^{(m)*}$ and $d^{(m)*}$ specified in \eqref{Eq:Optimal_Solution_Median} satisfy the following condition:
{
 \begin{align}
 \mathbb{F}=\left\{m\left\vert \min_{n\in\mathbb{N}}\left[r^{(m)}_n-b^{(m)*}\right]>0, \quad d^{(m)*}>0,\right.\right.\nonumber\\
  \left.\left.\boldsymbol{z}^{(m)*}_1\in\mathbb{R}, \quad m\in\mathbb{M}\right. \right\},
 \end{align}
where the first and second conditions mean that the distance estimation after deducing the bias and the step length estimation should be positive, and the third condition means that the coordinates of estimated position are real numbers.} The APs not in $\mathbb{F}$ are excluded and the relative trajectories $\{\mathcal{Z}^{(m)*}\}_{m\in \mathbb{F}}$ are used for the next~step. 

\begin{figure}[t]
	\centering
	\subfigure[Site A]{\includegraphics[width=7.5cm]{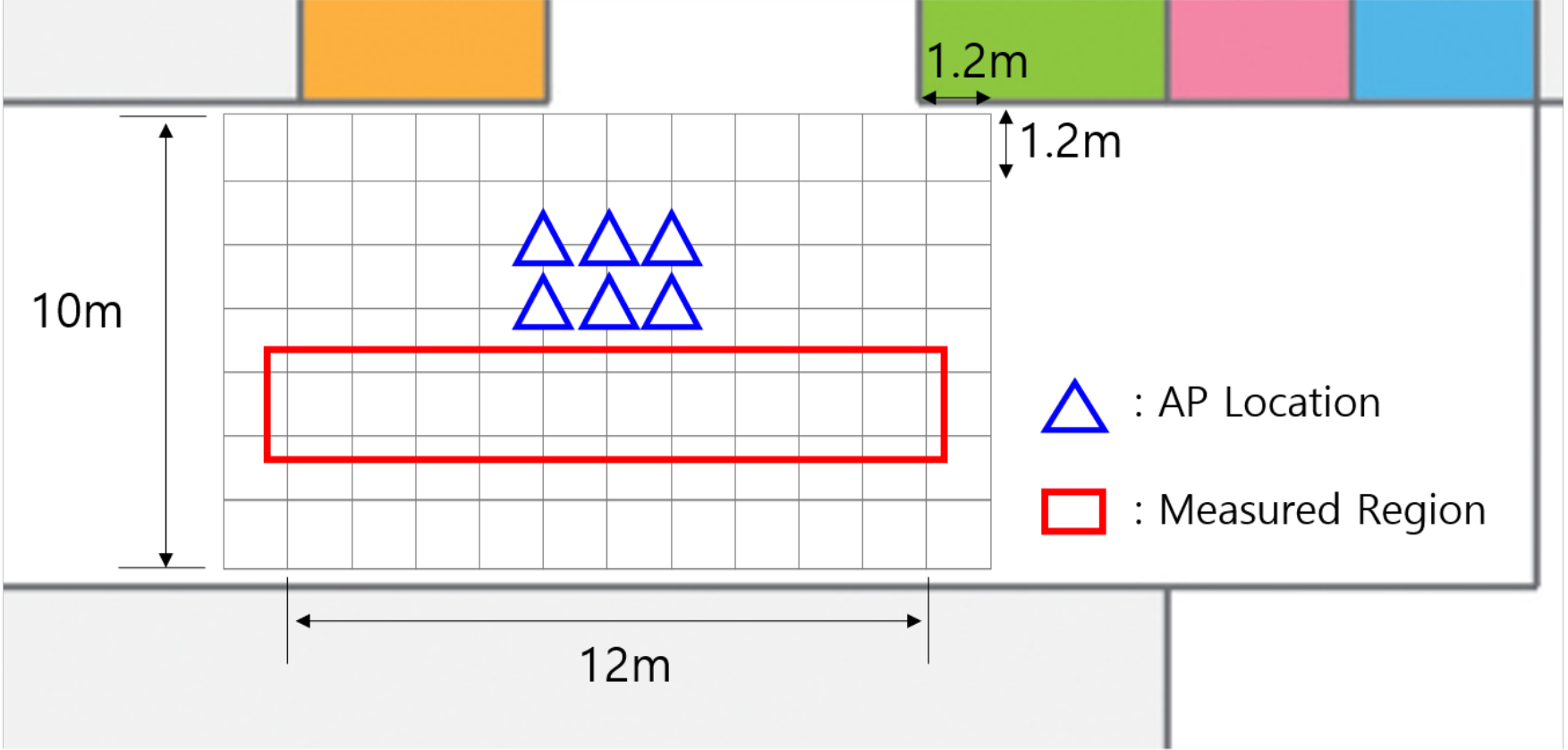}}
	\subfigure[Site B]{\includegraphics[width=8.5cm]{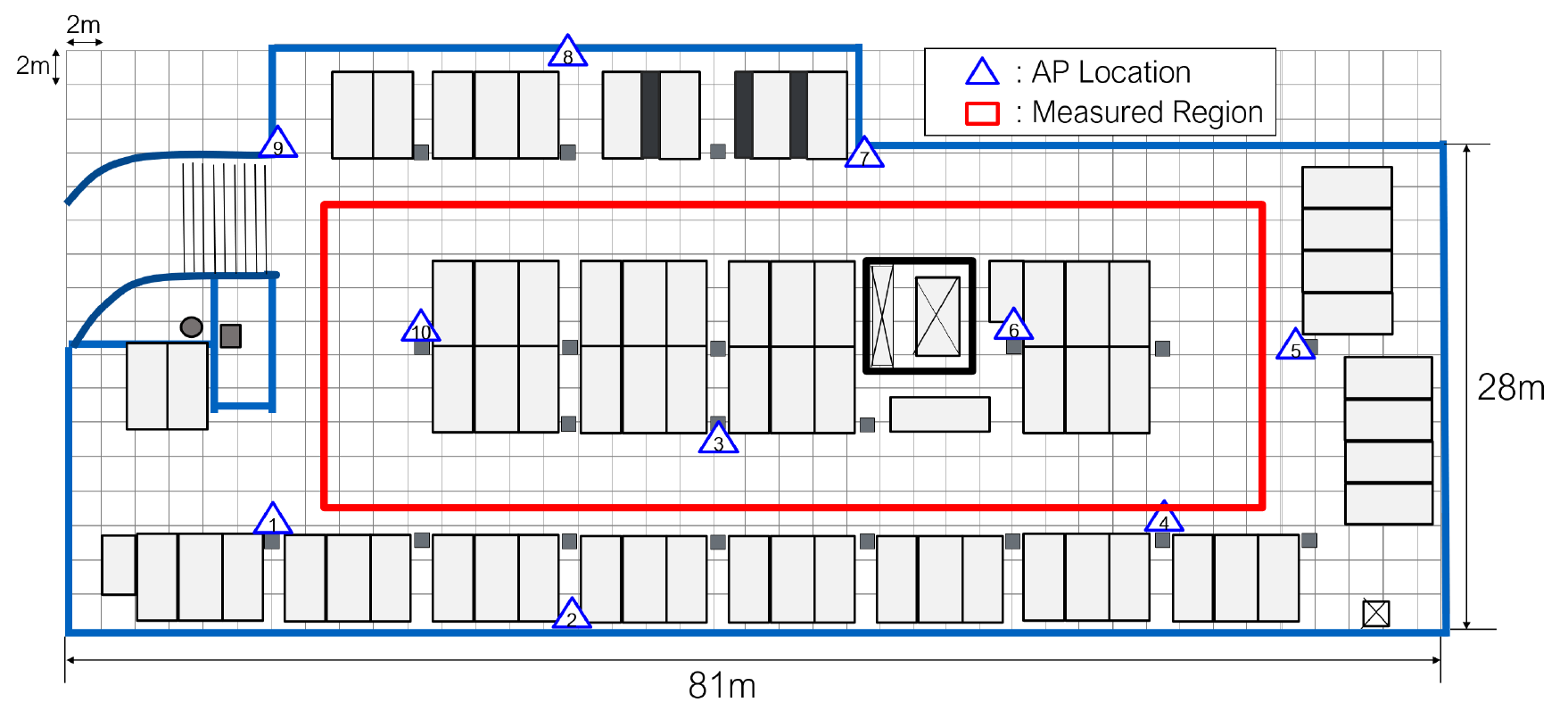}}
	\subfigure[{Site C}]{\includegraphics[width=7.5cm]{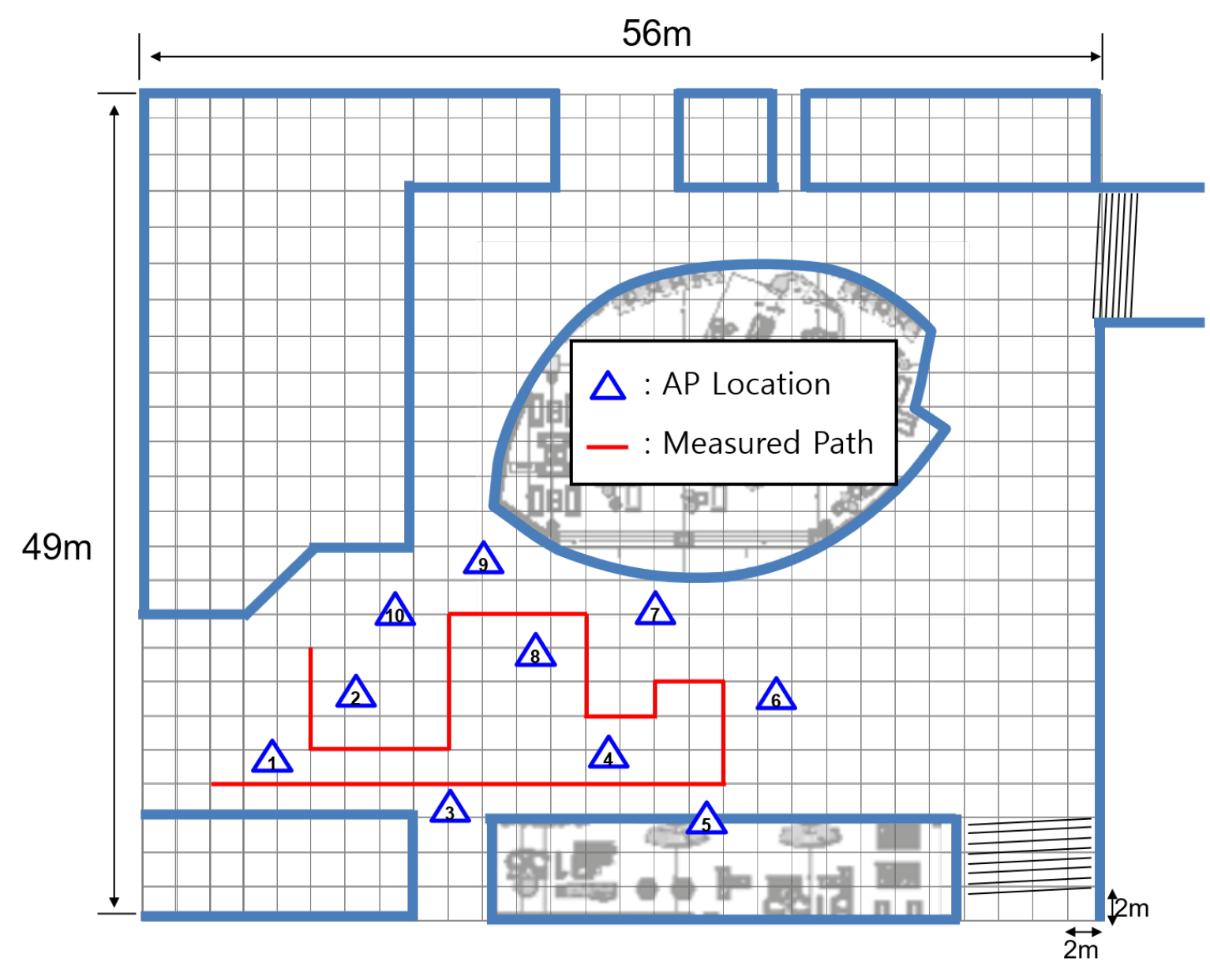}}
	\caption{Floor plans of field experiments. The locations of APs are represented as blue triangles. The user moves along each red line. The detailed explanations of each site and experiment setup are summarized in Table \ref{Tab:TestSetting}. }
	\label{Fig:TestEnvironments}
\end{figure}

\subsubsection{The Optimal Heading Direction Derivation} Depending on the mobility pattern being arbitrary or linear, different methods are used to find the optimal heading direction as follows.
\paragraph{Arbitrary mobility} 
Given $\mathbb{F}$, we aim at {aligning} all trajectories as closely as possible. Specifically, we define an error function $e_3(\omega)$ as the sum of the Euclidean distance between two relative trajectories in $\mathbb{F}$, given as
\begin{align}
	e_3(\omega)=\sum_{i,j\in\mathbb{F}}\sum_{n=1}^N\norm{\boldsymbol{p}_n^{(i)}(\omega)-\boldsymbol{p}_n^{(j)}(\omega)}.\nonumber
\end{align}
By a 1D search, it is possible to find $\omega^*$ to minimize the error function $e_3(\omega)$, namely,      
\begin{align} \label{Eq:TAErrorFunction}
	\omega^*=\arg\min_{\omega\in[0, 2\pi)} e_3(\omega).
\end{align}
\paragraph{Linear mobility} Recall that there {exists} the ambiguity of relative trajectories in the case of linear mobility, say $\{\mathcal{Z}_+^{(m)},\mathcal{Z}_-^{(m)} \}_{m\in \mathbb{F}}$ specified in \eqref{Eq:LocalTrajectoiesLinear}.
To remove this ambiguity, we utilize the relation between relative trajectories of different APs, as stated in the following proposition.
\begin{proposition}[The Relation Between Relative Trajectories]\label{Proposition4}\emph{
	The distance between relative positions of two APs is always equal to the distance between the two APs, namely, 
	\begin{align}
		&\norm{\boldsymbol{z}_n^{(m_1)*}-\boldsymbol{z}_n^{(m_2)*}}\nonumber\\
		=&\norm{\boldsymbol{p}_{\text{AP}}^{(m_1)}-\boldsymbol{p}_{\text{AP}}^{(m_2)}}, \quad \forall n\in \mathbb{N}, \quad \forall m_1, m_2\in \mathbb{M}.\nonumber
	\end{align}
}
\end{proposition}
\begin{proof}
	See Appendix \ref{Appen:Propostion4}. 
\end{proof}
Select one reference AP whose index is denoted by $r$. Depending on $\mathcal{Z}_{+}^{(r)}$ or $\mathcal{Z}_{+}^{(r)}$, there exist two possible sequences of relative trajectories, denoted by $\mathcal{Y}_+^{(r)}$ and $\mathcal{Y}_-^{(r)}$, initialized as $\{\mathcal{Z}_{+}^{(r)}\}$ and $\{\mathcal{Z}_{-}^{(r)}\}$, respectively. For example, given $\mathcal{Z}_{+}^{(r)}$, either \ref{Eq:ExampleOfPair1} or \ref{Eq:ExampleOfPair2} holds, given as
\begin{align} 
	\psi_1=\norm{\boldsymbol{p}_{\text{AP}}^{(r)}&-\boldsymbol{p}_{\text{AP}}^{(m)}}-\norm{\boldsymbol{z}_n^{(r)*}-\boldsymbol{z}_n^{(m)*}}=0, \nonumber\\
	& \forall n\in \mathbb{N}, 
	\quad \boldsymbol{z}_n^{(r)*}\in \mathcal{Z}_{+}^{(r)},  \quad \boldsymbol{z}_n^{(m)*}\in \mathcal{Z}_{+}^{(m)}, \label{Eq:ExampleOfPair1}\tag{C1} \\
	\psi_2=\norm{\boldsymbol{p}_{\text{AP}}^{(r)}&-\boldsymbol{p}_{\text{AP}}^{(m)}}-\norm{\boldsymbol{z}_n^{(r)*}-\boldsymbol{z}_n^{(m)*}}=0, \nonumber\\
& \forall n\in \mathbb{N}, \quad  \boldsymbol{z}_n^{(r)*}\in \mathcal{Z}_{-}^{(r)}, \quad \boldsymbol{z}_n^{(m)*}\in \mathcal{Z}_{+}^{(m)}. \label{Eq:ExampleOfPair2}\tag{C2}
\end{align}
When \ref{Eq:ExampleOfPair1} holds (i.e., $\psi_1=0$), $\mathcal{Z}_{+}^{(m)}$ and $\mathcal{Z}_{-}^{(m)}$ are added in $\mathcal{Y}_+^{(r)}$ and $\mathcal{Y}_-^{(r)}$, respectively. When \ref{Eq:ExampleOfPair2} holds (i.e., $\psi_2=0$), reversely, $\mathcal{Z}_{-}^{(m)}$ and $\mathcal{Z}_{+}^{(m)}$ are added in $\mathcal{Y}_+^{(r)}$ and $\mathcal{Y}_-^{(r)}$, respectively. The addition process is continued until $|\mathcal{Y}_+^{(r)}|=|\mathcal{Y}_-^{(r)}|=\mathbb{F}$.  In the presence of measurement noises, neither \ref{Eq:ExampleOfPair1} nor \ref{Eq:ExampleOfPair2} can be satisfied. Instead, we relax \ref{Eq:ExampleOfPair1} and \ref{Eq:ExampleOfPair2} as $\psi_1>\psi_2$ and $\psi_1\leq \psi_2$, respectively.

Given $\mathcal{Y}_+^{(r)}$ and $\mathcal{Y}_-^{(r)}$, relative trajectories are rotated using \eqref{Eq:GlobalPostion}, denoted by $\boldsymbol{p}_n^{(m)}\big(w;\mathcal{Y}_+^{(r)}\big)$ and $\boldsymbol{p}_n^{(m)}\big(w;\mathcal{Y}_-^{(r)}\big)$, respectively. 
We define $e_3^{(r)}\big(\omega\big)=\min\left\{e_3^{(r)}\big(\omega;\mathcal{Y}_+^{(r)}\big), e_3^{(r)}\big(\omega;\mathcal{Y}_-^{(r)}\big)\right\}$, where
\begin{align}
	e_3^{(r)}\big(\omega;\mathcal{Y}^{(r)}\big)
	=\sum_{i,j\in\mathbb{F}}\sum_{n=1}^N\norm[\Big]{\boldsymbol{p}_n^{(i)}\big(\omega;\mathcal{Y}^{(r)}\big)-\boldsymbol{p}_n^{(j)}\big(\omega;\mathcal{Y}^{(r)}\big)}, \nonumber\\
	 \mathcal{Y}^{(r)}\in\{\mathcal{Y}_+^{(r)},\mathcal{Y}_-^{(r)}\}.\nonumber
\end{align}
By a 1D search, it is possible to find the reference AP $r^*$ and the corresponding $\omega^*$ to 
minimize the error function as
\begin{align}
	\{r^*, \omega^*\} = \arg\min_{r\in\mathbb{F},  \omega\in[0,2\pi)} \left\{e_3^{(r)}\big(\omega\big) \right\}.\nonumber
\end{align}

\begin{figure}[t]
	\centering
	\subfigure[Sample trajectory]{\includegraphics[width=3cm]{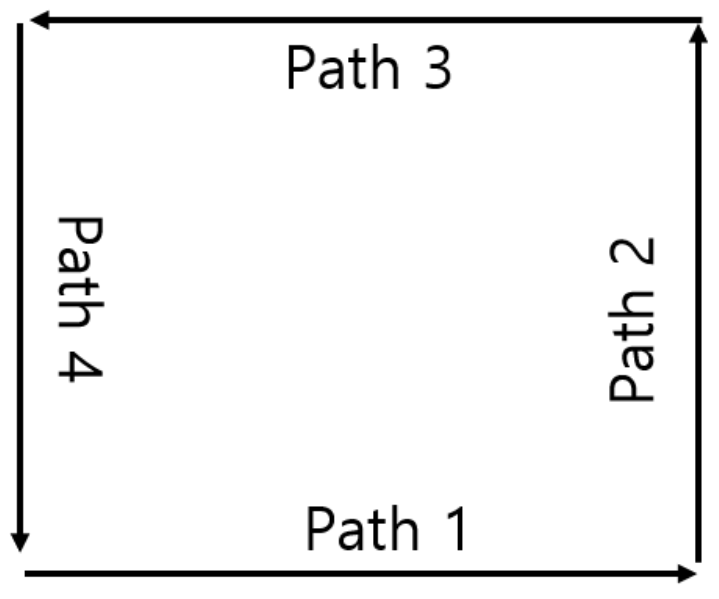}} 
	\subfigure[Quantization of heading change]{\includegraphics[width=8cm]{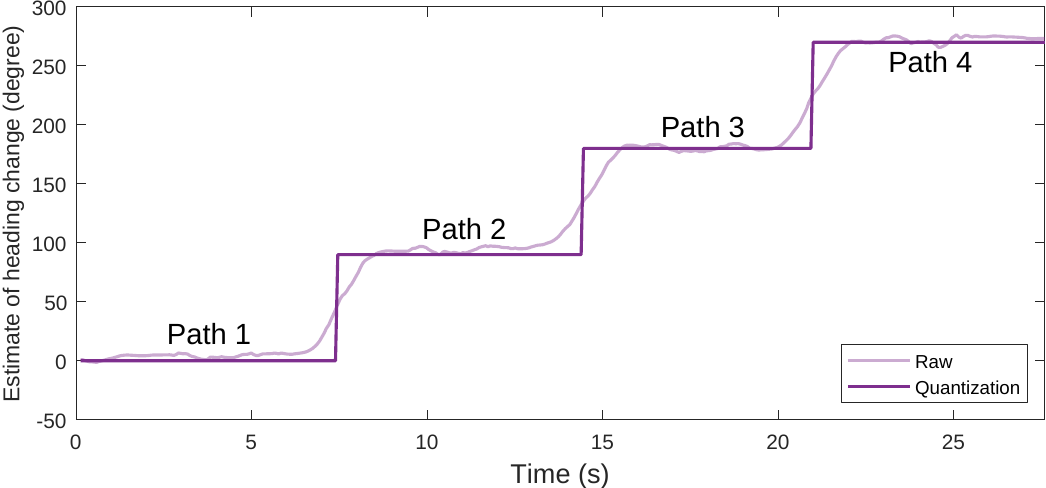}}
	\caption{{The graphical example of heading direction changes measured by a gyroscope. (a) Sample trajectory with four paths (b) the corresponding heading direction changes with and without quantization.} }
	\label{Fig:Quantization}
\end{figure}

\subsubsection{Determining the estimated trajectory}
Last, the estimated trajectory, say $\mathcal{P}=\{\boldsymbol{p}_n\}$, is derived by averaging 
$\{\boldsymbol{p}_n^{(m)}(\omega^*)\}_{m=1}^{|\mathbb{F}|}$ as
\begin{align} \label{Eq:GlobalPosition} 
	\boldsymbol{p}_n^*=\frac{1}{|\mathbb{F}|}\sum_{m\in\mathbb{F}} \boldsymbol{p}_n^{(m)}(\omega^*), \quad \forall n\in\mathbb{N}.
\end{align}

\section{Field Experiments}\label{Sec:Experiement}

\begin{figure*}[t]
	\centering
	\subfigure[Site A]{\includegraphics[width=7cm]{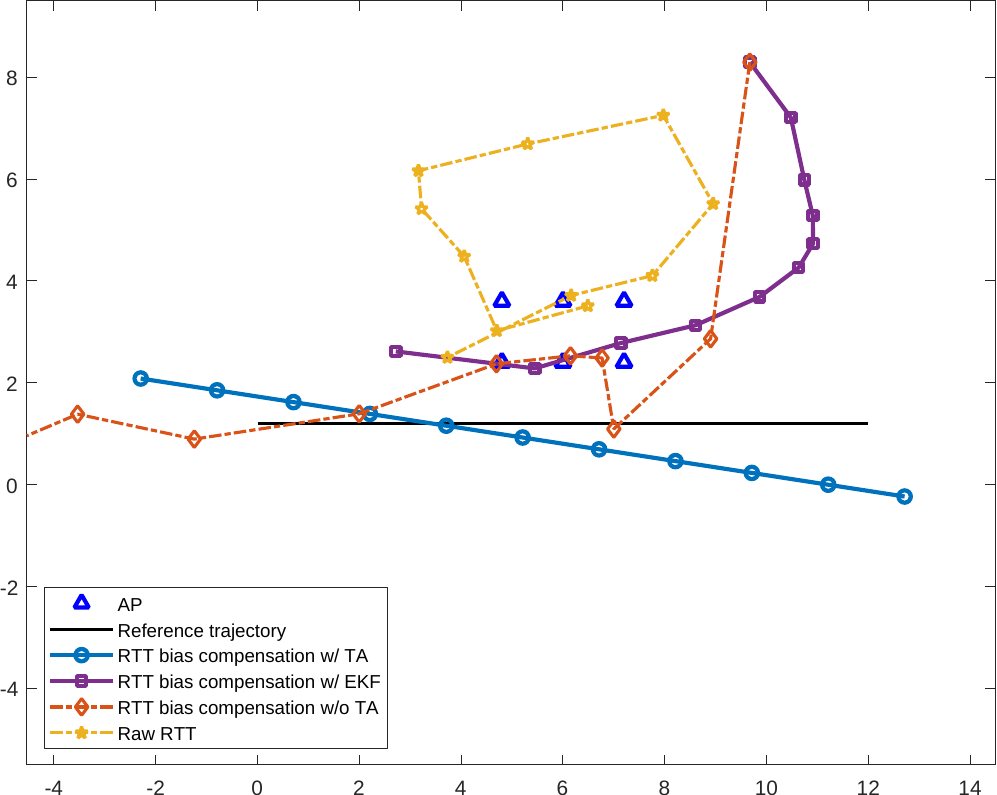}}
	\subfigure[Site B]{\includegraphics[width=7cm]{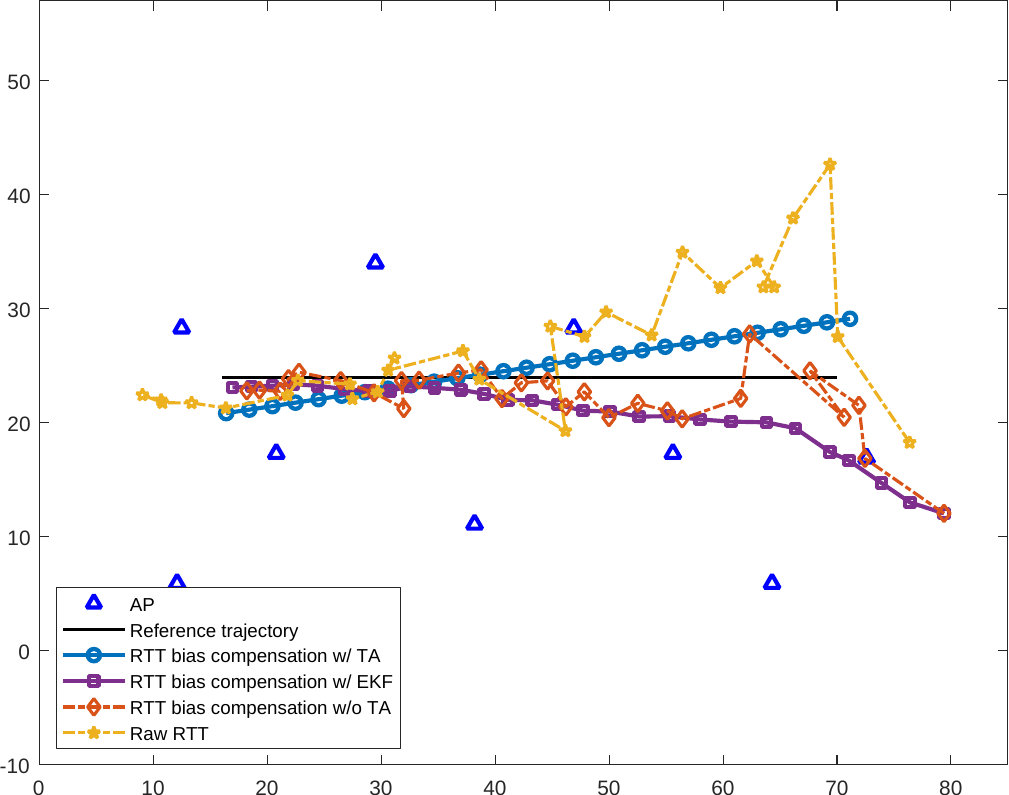}}
	\caption{The comparison between estimated and ground-truth trajectories: linear mobility cases. }
	\label{Fig:LinearPath}
\end{figure*}

\begin{figure*}[t]
	\centering
	\subfigure[Site A]{\includegraphics[width=7cm]{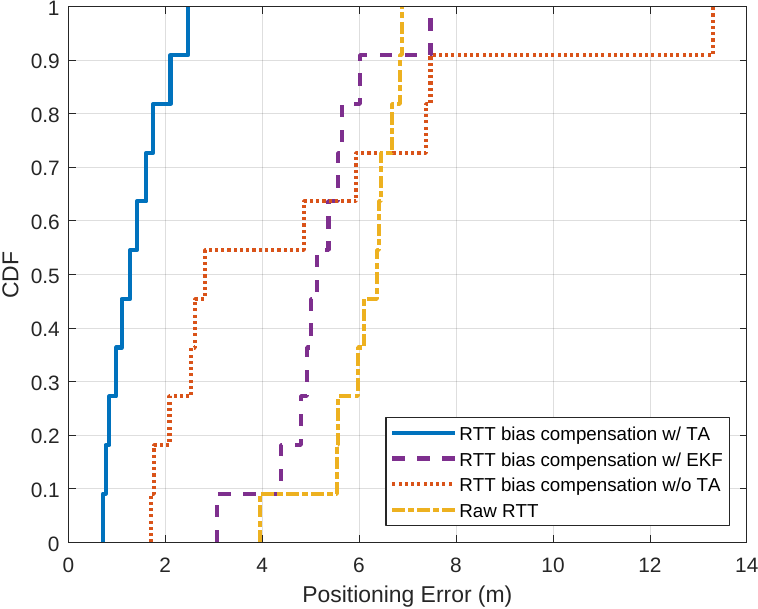}}
	\subfigure[Site B]{\includegraphics[width=7cm]{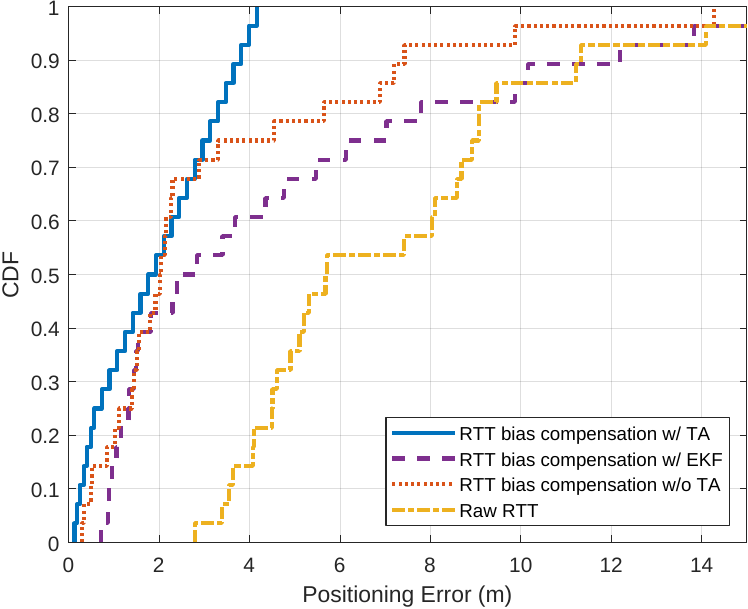}}
	\caption{The CDFs of positioning error for proposed algorithm and three benchmarks: linear mobility cases. }
	\label{Fig:LinearCDF}
\end{figure*}

This section aims at verifying the proposed positioning algorithms using several field experiments at two different indoor sites, each of which has different environments from the positioning perspective, as shown in Fig. \ref{Fig:TestEnvironments}. The first experimental site, called \emph{site A}, is a plaza located inside the Engineering building at Yonsei University, Seoul, Korea. Site A is an open space like outdoor environments where several LOS paths exist. On the other hand, the second experimental site, called \emph{site B}, is a parking lot located under {Building 11} in Korea Railroad Research Institute, Uiwang, Korea. {Compared with} site A, most signal propagations are followed by NLOS paths due to the presence of many obstacles such as {parked} vehicles, walls, and pillars. The detailed experiment setups are summarized in Table \ref{Tab:TestSetting}, unless specified.

We use the algorithm in \cite{Kang2015} to obtain the heading direction changes $\{\theta_n\}$ (see Fig. \ref{Fig:Quantization} as an example). 
As shown in Fig. \ref{Fig:Quantization}, we consider $\theta_n\in \{0$, $\frac{\pi}{2}$, $\pi$, $\frac{3\pi}{2}\}$, 
assuming that the user's moving direction only has finite choices depending on the surrounding arrangement (i.e., road, wall and et al.). Its effect is discussed in the sequel.

Three benchmarks are considered. The first one is based on raw RTT measurements without bias compensation for a conventional multilateration method, such as \emph{linear least square-reference selection} (LLS-RS) \cite{LLS_RS2008}. For the second one, compensated RTT measurements are utilized for a {multilateration} method, but TA is not applied. {The third one is an EKF-based algorithm in \cite{Sun2020}, which is operated based on $\delta^2$ and $\sigma^2$ representing the scaling factors of estimated positions from PDR and WiFi, respectively. We manually optimize these parameters and set them as $\delta^2=0.8$ and $\sigma^2=0.2$. For  a fair comparison, the RTT bias of each AP is compensated using our algorithm.
}
We use the Euclidean distance of estimated positioning to ground-truth locations to represent a positioning error, i.e.,  $\norm{\boldsymbol{p}_n^*-\boldsymbol{p}_n}$.

\subsection{Positioning Accuracy} \label{subsec:PositioningAccuracy}

We verify the performance of the algorithm for the cases of linear and arbitrary mobilities. 
The key performance metrics are summarized in Table \ref{Tab:TestResult}.

\subsubsection{Linear Mobility}

First, we consider the {cases} of linear mobility. Figure \ref{Fig:LinearPath} illustrates a graphical example of the estimated trajectories of the proposed one and three benchmarks, while Figure \ref{Fig:LinearCDF} shows  \emph{cumulative distributional functions} (CDFs) of the resultant positioning errors. 
Several key observations are made. First, the gain of the RTT bias compensation in Sec. \ref{Section:RTTRangingEnhancement} is verified by comparing two benchmarks, showing significant performance improvements for both Sites A and B. Second, TA explained in Sec. \ref{Sec:TrajectoryAlignment} makes all estimated points tailored to the trajectories detected by PDR, leading to additional performance enhancement. As a result, the resultant positioning errors of approximately $90\%$ are located within $2$ (m) and $4$ (m), and the average errors are $1.359$ (m) and $1.915$ (m) for Sites $A$ and $B$ respectively. 
The other performance metrics are summarized in Table \ref{Tab:TestResult}.

\subsubsection{Arbitrary Mobility}

\begin{figure*}[t]
	\centering
	\subfigure[Site A]{\includegraphics[width=5.5cm]{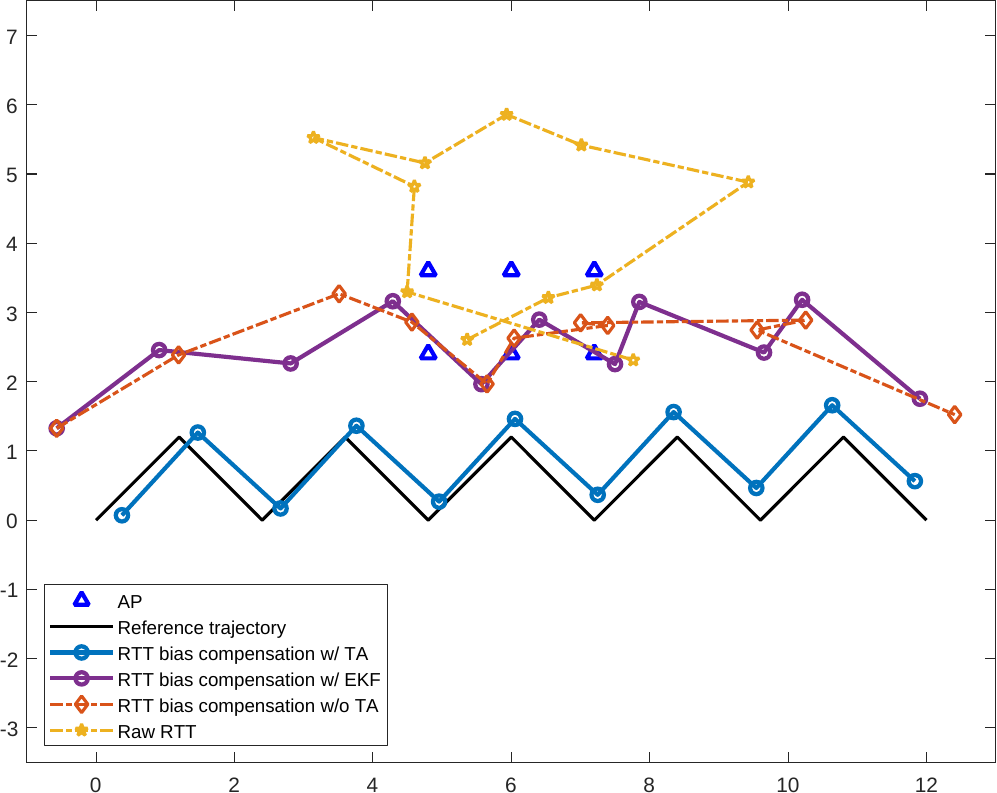}}
	\subfigure[Site B]{\includegraphics[width=5.5cm]{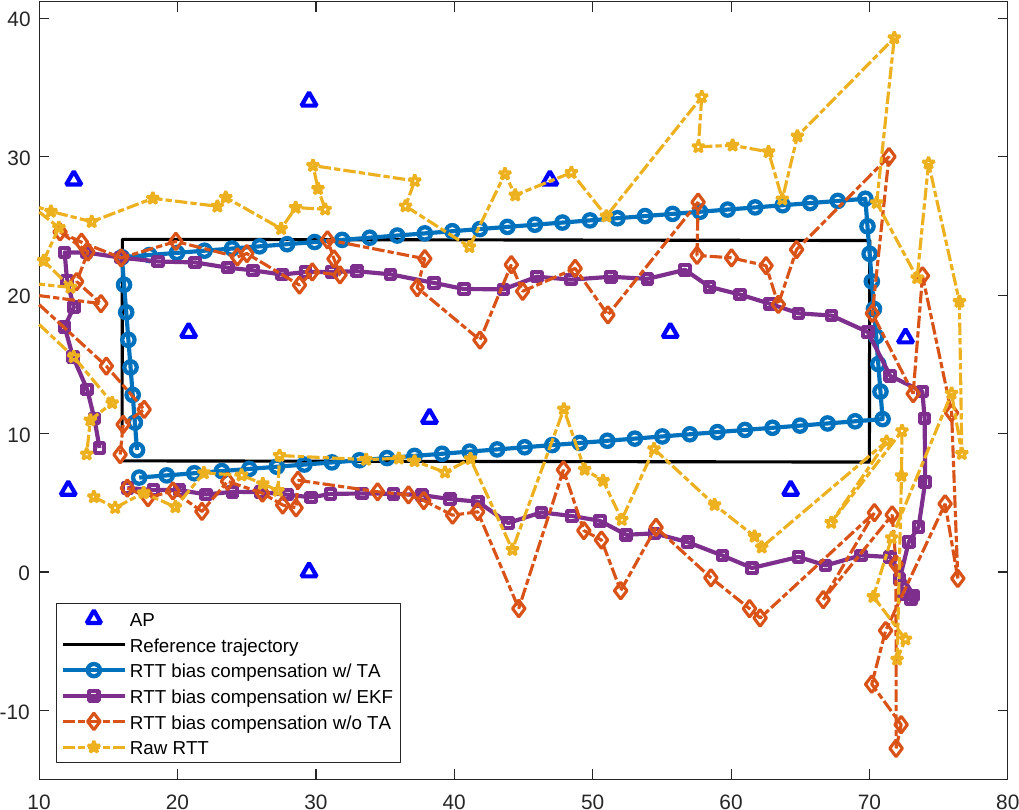}}
	\subfigure[{Site C}]{\includegraphics[width=5.5cm]{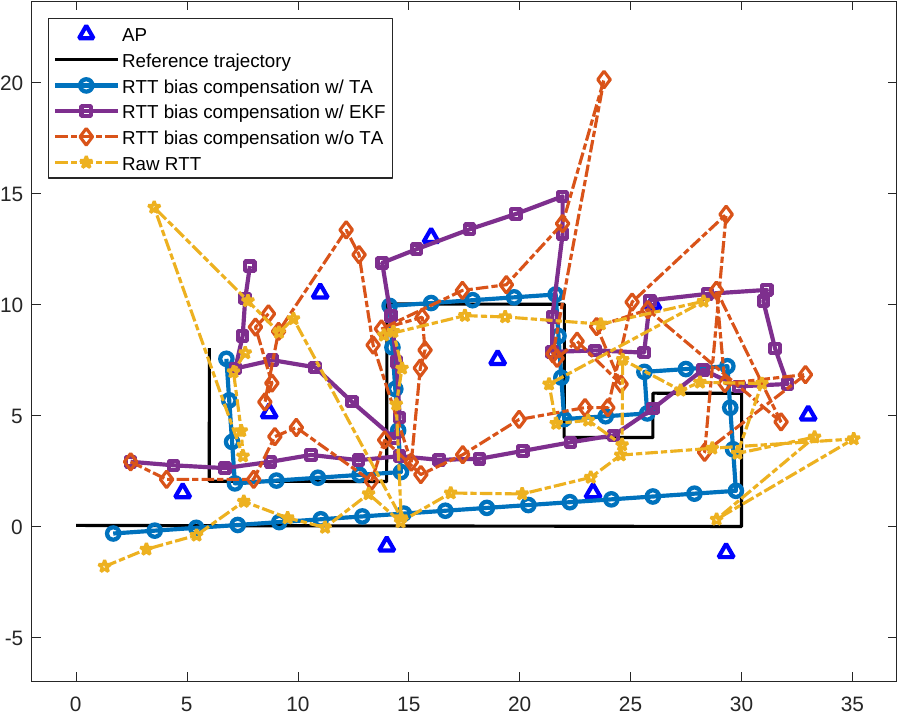}}
	\caption{The comparison between estimated and ground-truth trajectories: arbitrary mobility cases.}
	\label{Fig:ArbitraryPath}
\end{figure*}

\begin{figure*}[t]
	\centering
	\subfigure[Site A]{\includegraphics[width=5.5cm]{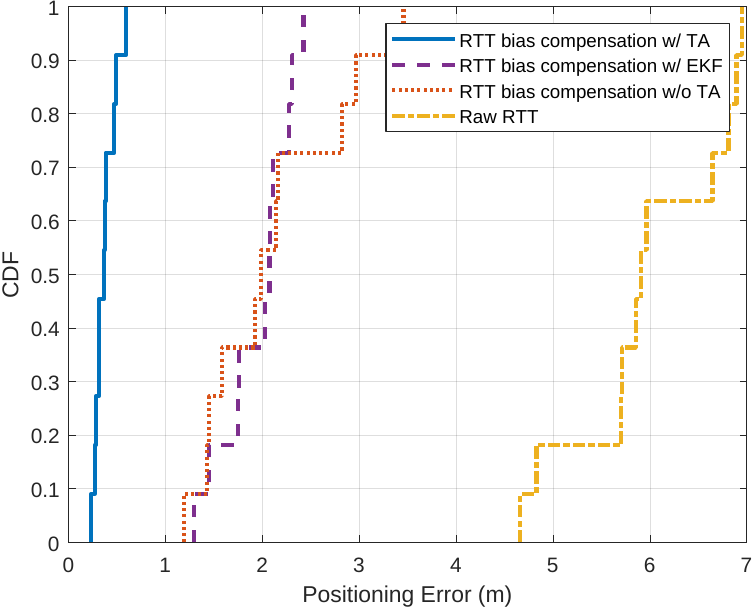}}
	\subfigure[Site B]{\includegraphics[width=5.5cm]{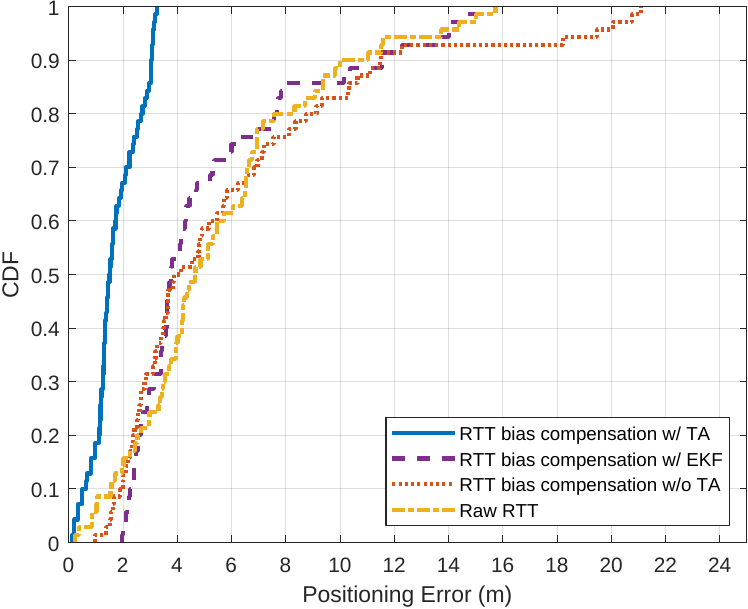}}
	\subfigure[{Site C}]{\includegraphics[width=5.5cm]{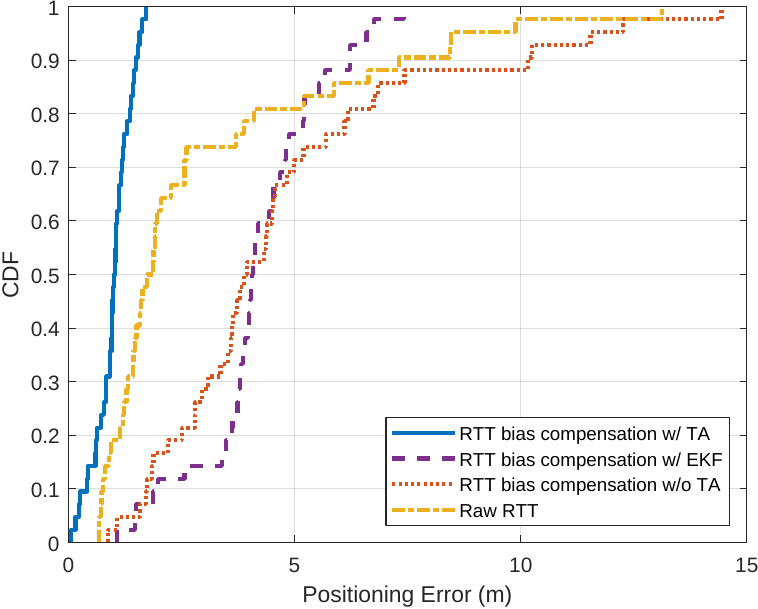}}
	\caption{The CDFs of positioning error for proposed algorithm and three benchmarks: arbitrary mobility cases.}
	\label{Fig:ArbitraryCDF}
\end{figure*}

Second, we consider the {cases} of arbitrary mobility.  we illustrate a graphical example of the estimated trajectories in Figure \ref{Fig:ArbitraryPath} and \emph{cumulative distributional functions} (CDFs) of the resultant positioning errors in  Figure \ref{Fig:ArbitraryCDF}, showing similar tendencies to the linear mobility counterpart. Besides, it is shown that the algorithm for arbitrary mobility provides a more accurate positioning result than that for linear mobility such that 
{the resultant positioning errors of $90\%$ are approximately located within $0.5$ (m), $2.5$ (m), and $1.5$ (m), and the average errors are $0.369$ (m), $1.705$ (m), and $0.978$ (m) for Sites $A$, $B$, and $C$, respectively.} 

{It is noteworthy that the algorithm proposed for arbitrary mobility provides a more accurate positioning result than the linear mobility counterpart (see Table \ref{Tab:TestResult}). As recalled in Remark 6, arbitrary mobility embeds 2D location information, while linear mobility embeds 1D information. The difference results in a significant accuracy difference between the two.}

{
\subsubsection{Comparison with EKF-based Algorithm}
As shown in \Cref{Fig:LinearPath,Fig:LinearCDF,Fig:ArbitraryPath,Fig:ArbitraryCDF} and Table \ref{Tab:TestResult}, we confirm that our proposed algorithm outperforms the EKF-based one for both linear and arbitrary mobility scenarios. The EKF-based algorithm updates the current position based on the previously estimated ones, resulting in accumulating errors on positioning decisions that have been already made. On the other hand, the proposed one can correct such errors by utilizing the full trajectory information.
}

\subsection{Effect of Weighed Factors $w_1$ and $w_2$}

\begin{figure}[t]
	\centering
	\includegraphics[width=6cm]{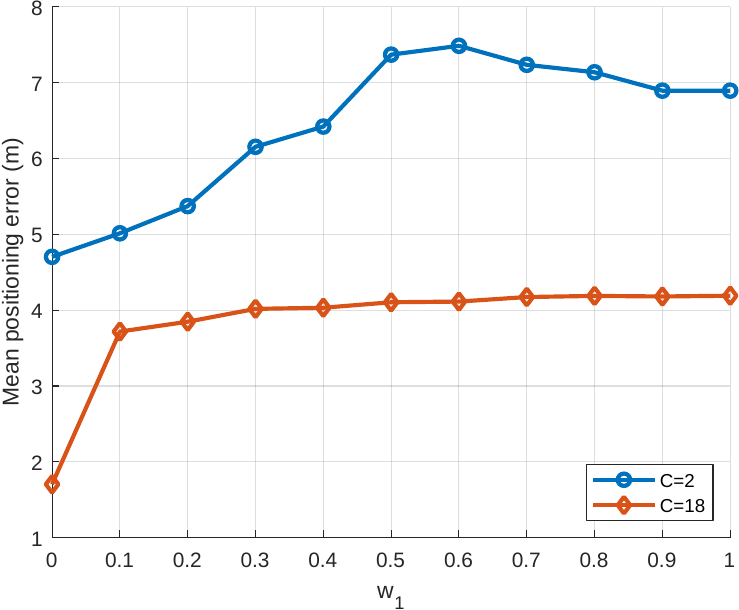}
	\caption{The effect of weight factors $\{w_1, w_2\}$ on positioning error. Site B is considered.}
	\label{Fig:WeighedTest}
\end{figure}
Recall a pair of $w_1$ and $w_2$ weighting the error functions $e_1$ of \eqref{Eq:Cost_Function_1} and $e_2$ of \eqref{Eq:Cost_Function_2}, respectively. As stated in Remark \ref{Remark:EffectOfErrorFunction2},  the error function $e_1$ focuses on minimizing the error on step length $d$ and the RTT bias $b$, while $e_2$ {aims} at minimizing the error on the initial position $\boldsymbol{z}_1^*$. Fig. \ref{Fig:WeighedTest} represents the effect of weight factors on the average positioning error,
showing that $\{w_1,w_2\}=\{0,1\}$ provides the most accurate positioning result. It is because the errors on $d$ and $b$ can be compensated by the method of multiple combination of RSs introduced in Sec.  \ref{Subsec:RS_Selection}. Therefore, if a sufficient number of RSs is selected, it is optimal to focus on the minimization of $e_2$.

\subsection{Effect of Multiple Combinations of RSs}
\begin{figure}[t]
	\centering
	\includegraphics[width=6cm]{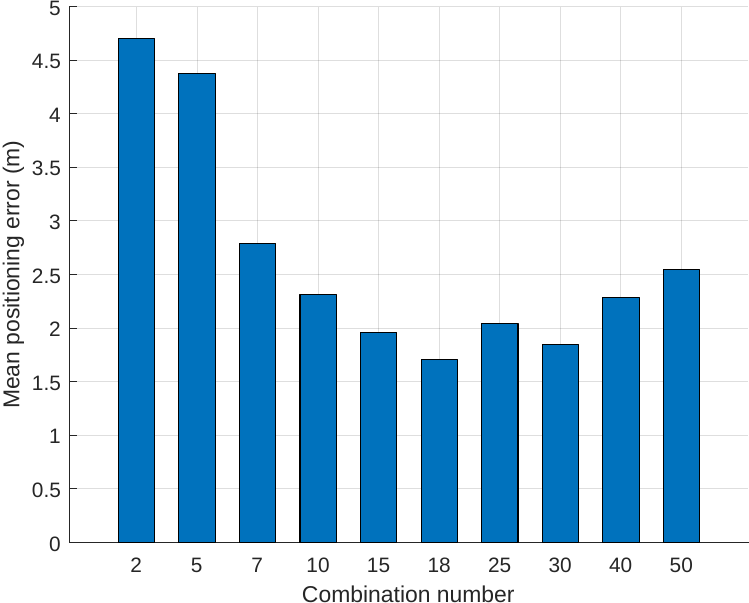}
	\caption{The positioning performance according to the combination number in Site B. $C=18$ has the best performance.}
	\label{Fig:CombinationTest}
\end{figure}

Fig. \ref{Fig:CombinationTest} shows the average positioning error as a function of the number of candidate RSs $C$, showing that the error is reduced from {$4.702$ (m) to $1.705$ (m)}, when the number of RSs $C$ increases from  $2$ {to} $18$. On the other hand, $C$ larger than $18$ rather deteriorates a positioning result since a larger portion of RSs is likely to be severely corrupted by the measurement error. Through various simulation studies, the positioning error can be {minimized in average sense} by selecting the number of RSs $C$ as $0.25 N$, where $N$ is the number of walking steps 
(e.g., $\frac{C}{N}=\frac{18}{70}\approx 0.257$ in the experiment setting of Fig.~\ref{Fig:CombinationTest}) . 

\subsection{Effect of Heading Direction Error}
\begin{figure}[t]
	\centering
	\subfigure[Arbitrary: Site A]{\includegraphics[width=8cm]{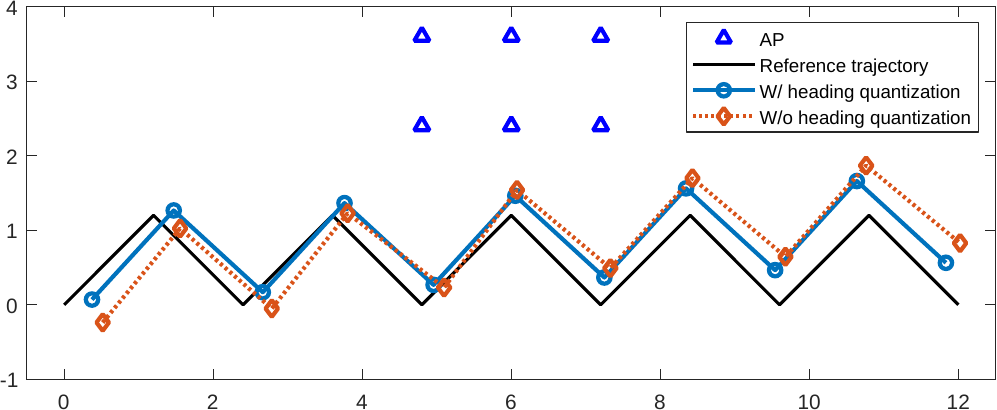}}
	\subfigure[Arbitrary: Site B]{\includegraphics[width=8cm]{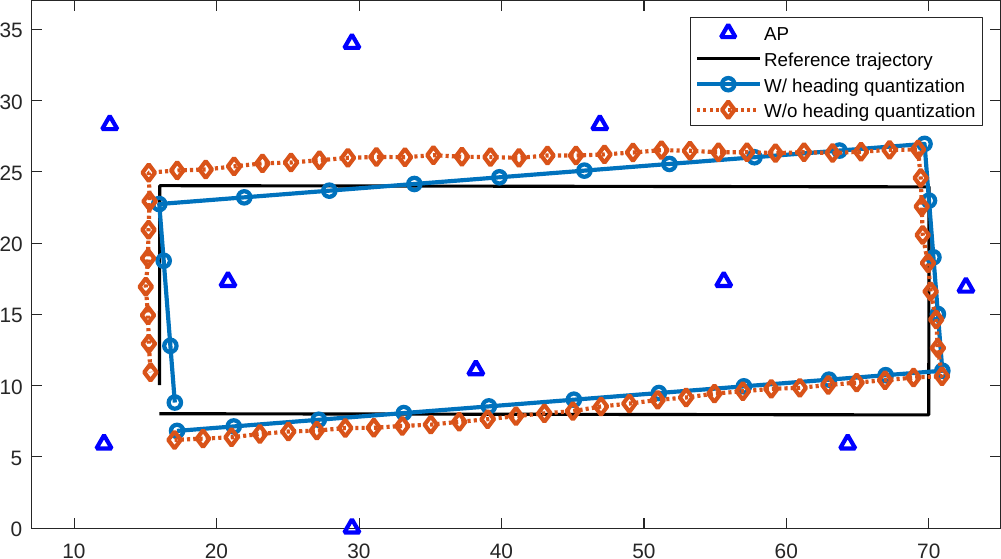}}
	\subfigure[{Arbitrary: Site C}]{\includegraphics[width=8cm]{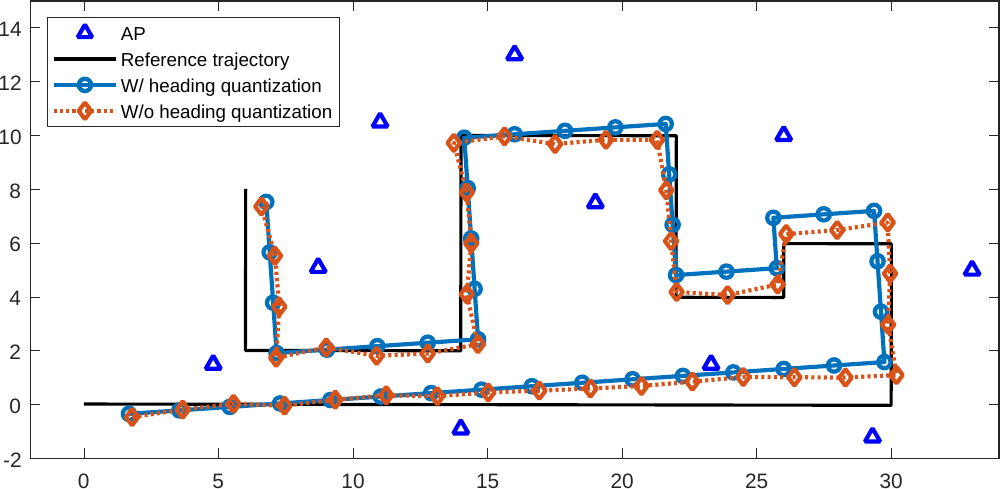}}
	
	\caption{The comparison between estimated trajectories with and without the quantization of heading direction change $\{\theta_n\}$. }
	\label{Fig:PriorPath}
\end{figure}

{
Recall that heading direction changes $\{\theta_n\}$ are quantized into four levels as $\theta_n\in \{0$, $\frac{\pi}{2}$, $\pi$, $\frac{3\pi}{2}\}$, by exploiting the prior information of surrounding arrangement. It enables us to obtain the precise trajectory estimation, as shown in Fig. \ref{Fig:LinearPath} and \ref{Fig:ArbitraryPath}. 
On the other hand, Fig.~\ref{Fig:PriorPath} plots the estimated  trajectory without quantization, showing that the degradation of the positioning accuracy is marginal such that the average positioning error increases from $0.369$ (m) to $0.494$ (m) at Site A, $1.705$ (m) to $1.958$ (m) at Site B.  At Site C, the positioning accuracy is slightly improved from $0.978$ (m) to $0.873$ (m). With the gyroscope error being unbiased, each error can be canceled out while moving. As a result, the entire trajectory is generally well-maintained with acceptable error, verifying the proposed algorithm to be robust against such error. }

\begin{table}
	\caption{Experiment details} \label{Tab:TestSetting}
	\begin{adjustbox}{width=\textwidth/2}
		\begin{tabular}{l|lll}
			& Site A  & Site B & Site C \\ \hline
			\multirow{2}{*}{Address} & Yonsei University,  & Korea Railroad Research & Yonsei University  \\
			& Seoul, Korea & Institute, Uiwang, Korea & Seoul, Korea \\ [0.5em]
			Type of site & Indoor, plaza  & Underground, parking lot & Indoor, conference hall \\ [0.5em]
			\multirow{2}{*}{Chip} & Intel Dual Band   &  Qualcomm IPQ 4018 & Qualcomm IPQ 4018 \\
			& Wireless AC 8260 &   & \\ [0.5em]
			Bandwidth& $80$ MHz  & $40$ MHz &$40$MHz  \\ [0.5em]
			Carrier & $5.24$ GHz  & $5.18$ GHz & $5.18$ GHz \\ [0.5em]
			Number of WiFi APs & $6$   & $10$ & $10$  \\ [0.5em]
			Heights of APs & $1.8$ m  & $2.0$ m & $2.0$m  \\ [0.5em]
			Heights of mobile & $1.8$ m  & $1.1$ m  & $1.1$m \\ [0.5em]
			Smartphone & Google Pixel $2$XL  & Google Pixel $2$XL & Google Pixel $2$XL   \\ [0.5em]
			Version & Android $9$  & Android $9$ & Android $9$ \\  [0.5em]  
			Weight factors $\{w_1, w_2\}$  & \{0,1\}  & \{0,1\} & \{0,1\} \\ [0.5em] 
			\multirow{2}{*}{Number of candidate RSs $C$}  & Linear: $2$  & Linear: $5$ & \multirow{2}{*}{Arbitrary: $12$} \\
			&Arbitrary: $3$ & Arbitrary: $18$
			\\\hline
		\end{tabular}
	\end{adjustbox}
\end{table}

\begin{table}[t]
	\centering
	\caption{Several performance metrics of each algorithm.}\label{Tab:TestResult}
	\begin{adjustbox}{width=\textwidth/2}
		\begin{tabular}{cccl|ccccc} 
			\toprule
			\multirow{2}{*}{Site}& Mobility &  Number  & \multirow{2}{*}{Algorithm}   & \multicolumn{5}{c}{Positioning Error (m)}  \\ 
			\cline{5-9}
			& type & of steps &      & min   & max   & mean  & med   & std     \\  
			\hhline{---------}
			\multirow{3}{*}{A} &\multirow{3}{*}{Linear} & \multirow{3}{*}{11}& w/ TA       & 0.705 & 2.459 & 1.359 & 1.269 & 0.570   \\
			&& & w/ EKF        & 3.063 & 7.468& 5.211& 5.128 & 1.081   \\
			&& & w/o TA        & 1.696 & 13.292& 4.762 & 2.808 & 3.569   \\
			&& & Raw RTT             & 3.95 & 6.874 & 6.064 & 6.365 & 0.835    \\ 
			\hline
			\multirow{3}{*}{B} &\multirow{3}{*}{Linear} & \multirow{3}{*}{28} &w/ TA           & 0.124 & 4.157 & 1.915 & 1.842 & 1.307   \\
			&& & w/ EKF        & 0.703 & 15.196 & 4.496& 2.612 & 4.249   \\
			&& & w/o TA         & 0.286 & 14.270 & 3.182 & 2.031 & 3.281   \\
			&& & Raw RTT              & 2.797 & 18.960 & 5.692 & 5.692 & 3.664    \\ 
			\hline
			\multirow{3}{*}{A} &\multirow{3}{*}{Arbitrary} &\multirow{3}{*}{11} & w/ TA           & 0.231 & 0.588 & 0.369 & 0.365 & 0.108   \\
			&& & w/ EKF        & 1.290 & 2.422 & 1.956& 2.071 & 0.358   \\
			&& & w/o TA              & 1.187 & 3.454 & 2.099 & 1.986 & 0.716   \\
			&& & Raw RTT              & 4.660 & 6.945 & 5.992 & 5.910 & 0.783    \\ 
			\hline
			\multirow{3}{*}{B} & \multirow{3}{*}{Arbitrary}& \multirow{3}{*}{70}& w/ TA             & 0.121 & 3.245 & 1.705 & 1.504 & 0.865  \\
			&& & w/ EKF        & 1.947 & 15.115 & 5.243& 3.773 & 3.510   \\
			&& & w/o TA       & 0.976 & 21.092 & 5.963 & 4.015 & 4.866   \\
			&& & Raw RTT             & 0.218 & 15.727 & 5.525 & 4.653 & 3.555    \\ 
			\hline
			\multirow{3}{*}{C} & \multirow{3}{*}{Arbitrary}& \multirow{3}{*}{70}& w/ TA             & 0.046 & 1.715 & 0.978 & 1.014 & 0.407  \\
			&& & w/ EKF        & 1.064 & 7.406 & 4.199& 4.064 & 1.397   \\
			&& & w/o TA       & 0.857 & 14.439 & 4.743 & 3.095 & 3.072   \\
			&& & Raw RTT             & 1.064 & 7.406 & 4.199 & 4.064 & 1.397    \\ 
			\bottomrule
		\end{tabular}
	\end{adjustbox}
\end{table}

\section{Concluding Remark}\label{Sec:Conclusion}

This work has presented a novel positioning algorithm to estimate a user's location by integrating RTT and PDR measurements. Geometric relations between the two are formulated as mathematical form, enabling us to design a tractable and scalable  positioning algorithm as well as provide the feasible conditions of the number of steps and WiFi APs for a unique positioning. The superiority of the proposed method has been well verified by field experiments that the positioning accuracy can be significantly improved than the conventional multilateration techniques. 

{
The current work can be extended in several directions. First, our algorithm can help \emph{simultaneous localization and mapping} (SLAM) by identifying whether a blockage exists from the concerned AP from the corresponding RTT bias. Second, our algorithm can be applied to the area of vehicle and drone positioning \cite{Ko2021} that has not been actively explored yet. Last, it is interesting to integrate such a sensing system into a wireless communication system using several techniques, e.g., compressive sensing \cite{Han2019b} and over-the-air computation \cite{Chen2018}, a key direction in B5G/6G communications.
}

\appendix
\section{Appendix} \label{Appendix}

\begin{figure*}
	\centering
	\begin{minipage}{0.45\textwidth}
		\centering 
		\includegraphics[height=6.5cm]{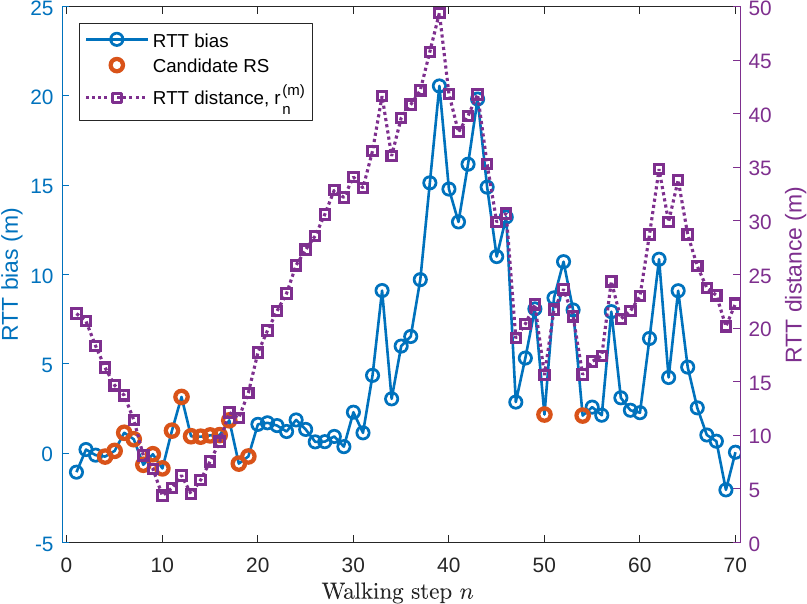}
		\caption{{The trace of RTT and bias measured by AP 3 in Site B. Candidate RSs are indicated as red dots.}}
		\label{Fig:Bias_and_RTT}
	\end{minipage}\hfill
	\begin{minipage}{0.48\textwidth}
		\centering 
		{\includegraphics[height=6cm]{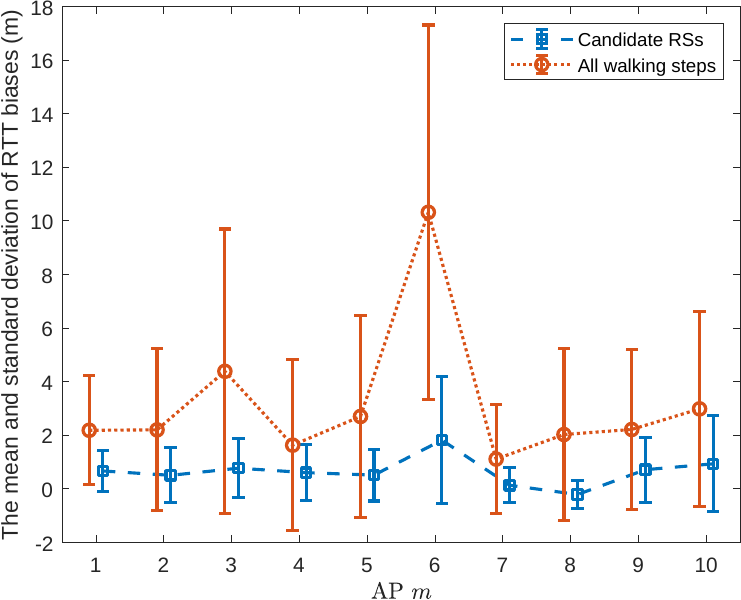}}
		\caption{{The mean and standard deviation of RTT biases in Site B. The sets of candidate RSs and all walking steps in the trajectory are considered.}}
		\label{Fig:RTTStd}
	\end{minipage}
\end{figure*}

\subsection{Verifying Effect of Deterministic Bias Assumption}
\label{Subsec:EffectofBias}

%

{
	This subsection aim at verifying that our algorithm can reduce the error due to the deterministic bias in Assumption 1 as marginal using RS selection and multiple combinations of RSs. 
	\begin{itemize}
		\item \textbf{Reference Step Selection}: Recall that when making the systems of linear equations in \ref{Eq:Linearization_Linear} and \ref{Eq:Linearization_Arbitrary}, two RSs play a pivotal role in the linearization process (see Sec. \ref{Subsec:LinearMobility} and Sec. \ref{Subsec:ArbitraryMobility}). In other words, with the RSs whose biases violate Assumption \ref{Assumption:Bias}, their  effects propagate throughout the entire equations, causing severe errors on the resulting solution. It is thus required to select RSs with biases being relatively constant.  
		Fig. \ref{Fig:Bias_and_RTT} plots the variations of RTT and bias during the full-trajectory in Site B, showing that the biases within a specific range are relatively constant when the corresponding RTTs are small.  As a result, we select two RSs whose RTTs are smaller.  Compared with selecting RSs randomly, selecting two RSs having the smallest RTT can reduce the average error on step length estimation up to $0.231$ (m).  
		\item \textbf{Multiple Combinations of Reference Steps}: Fig. \ref{Fig:StepLengthError}(a) represents the estimated step length  from two RSs sorted by RTT measurements; A smaller index represents that the sum of the selected RSs' RTTs is smaller. It is observed that smaller RTT does not always provide a higher estimation accuracy. To cope with this issue, several candidate RSs are selected instead of a single pair (See Sec. \ref{Subsec:RS_Selection}). Fig.~\ref{Fig:RTTStd} shows that the standard deviation of candidate RSs' RTTs  is significantly smaller than that of all RTTs, validating the deterministic bias assumption in Assumption \ref{Assumption:Bias}.  It leads to providing a reliable step-length estimation for all APs as shown in Fig. \ref{Fig:StepLengthError}(b). 
	\end{itemize}
}

\begin{figure*}[t] 
	\centering
	\subfigure[{Step length estimation of different pairs of RSs for AP 1.} ]{\includegraphics[width=7cm]{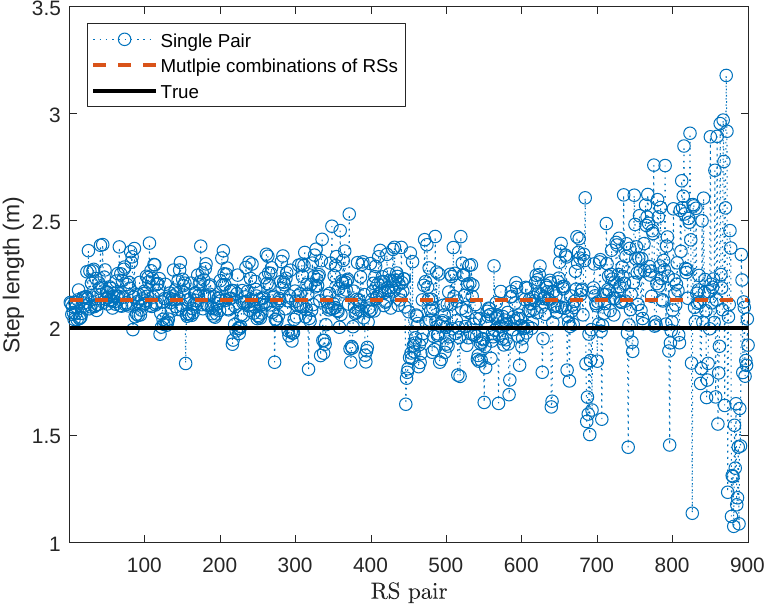}}
	\subfigure[{The average step length estimation error for every AP.} ]{\includegraphics[width=7cm]{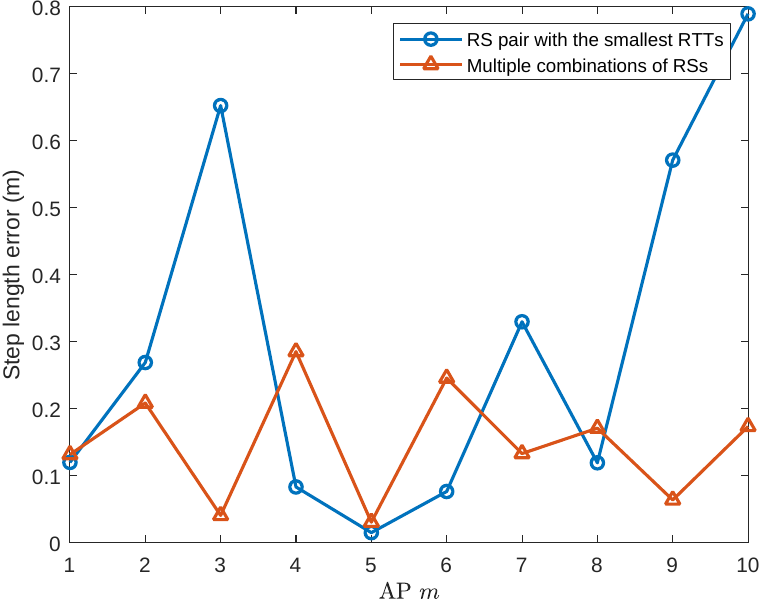}}
	\caption{{Step length estimation in Site B.}}
	\label{Fig:StepLengthError}
\end{figure*}

\subsection{Proof of Proposition \ref{Proposition2}} \label{Appen:Propostion2}
Noting that \ref{Eq:Linearization_Arbitrary} has $(N-2)$ equations, it is an overdetermined system when $N\geq5$ if the rank of the matrix $\boldsymbol{A}(\mathbb{S},\gamma)$ in \eqref{Eq:A_B_Arbitrary} is $3$. Then, there always exists $\gamma^*$ satisfying 
$e_1(\mathbb{S},\gamma^*)=0$ since \ref{Eq:Linearization_Arbitrary} is formulated from a geometric representation of multiple unknowns.  

Consider a special case of $\gamma=\hat{\gamma}$ where $f_{a_1,a_2}(\hat{\gamma})=0$. At this case,  $f_{n,a_2}(\hat{\gamma})=f_{n,a_1}(\hat{\gamma})$ because if $a_1<a_2$, $f_{n,a_2}(\hat{\gamma})= \sum_{i=1}^{n-1}\cos(\theta_{i}-\hat{\gamma})-\big(\sum_{i=1}^{a_1-1}\cos(\theta_{i}-\hat{\gamma})+\sum_{i=a_1}^{a_2-1}\cos(\theta_{i}-\hat{\gamma})\big ) $ and $f_{a_1,a_2}(\hat{\gamma})=\sum_{i=a_1}^{a_2-1}\cos(\theta_{i}-\hat{\gamma})$ by definition. From this, the matrix components $\alpha_n(\mathbb{S},\hat{\gamma})$ and $\beta_n(\mathbb{S},\hat{\gamma})$ can be changed as
\begin{align} \label{Eq:AlphaChanged}
	\alpha_n(\mathbb{S},\hat{\gamma})&=f_{n,a_2}(\hat{\gamma})\eta_{n,a_1}-f_{n,a_1}(\hat{\gamma})\eta_{n,a_2} \nonumber \\
	&=f_{n,a_1}(\hat{\gamma})(\eta_{n,a_1}-\eta_{n,a_2}) \nonumber \\
	&=f_{n,a_1}(\hat{\gamma})\eta_{a_2,a_1},
\end{align} 
\begin{align} \label{Eq:BetaChanged}
	\beta_n(\mathbb{S},\hat{\gamma})&=2\big(f_{n,a_2}(\hat{\gamma})(r_n-r_{a_1})-f_{n,a_1}(\hat{\gamma})(r_n-r_{a_2})\big) \nonumber \\
	&=2\big(f_{n,a_1}(\hat{\gamma})(r_{a_2}-r_{a_1})\big).
\end{align} 
Above \eqref{Eq:AlphaChanged} and \eqref{Eq:BetaChanged} have common terms $f_{n,a_1}(\hat{\gamma})$, the matrix $\boldsymbol{A}(\mathbb{S}, \hat{\gamma})$ which defined as \eqref{Eq:A_B_Arbitrary} becomes
\begin{align}
	\boldsymbol{A}(\mathbb{S}, \hat{\gamma})=
	\begin{bmatrix}
		\eta_{a_2,a_1} & 2(r_{a_2}-r_{a_1})
	\end{bmatrix}.
\end{align}
Therefore, the rank of $\boldsymbol{A}(\mathbb{S}, \hat{\gamma})$ is 1, it is underdetermined system. This finishes the proof.

\subsection{Proof of Proposition \ref{Proposition3}} \label{Appen:Propostion3}

We provide separated proofs for arbitrary and linear mobilities. 

\subsubsection{Arbitrary mobility} Consider the case with APs $1$ and $2$. 
Using \eqref{Eq:GlobalPostion},  a pair of relative locations at $n=1$, say $\boldsymbol{z}_1^{(1)*}$ and $\boldsymbol{z}_1^{(2)*}$, are given as	
\begin{align} \label{Eq:z_n^1}
	\boldsymbol{z}_1^{(1)*}=
	\boldsymbol{\Theta}(-\omega) \big(\boldsymbol{p}_1^{(1)}(\omega)-\boldsymbol{p}_{\text{AP}}^{(1)}\big), 
	\\
	\boldsymbol{z}_1^{(2)*}=
	\boldsymbol{\Theta}(-\omega) \big(\boldsymbol{p}_1^{(2)}(\omega)-\boldsymbol{p}_{\text{AP}}^{(2)}\big), 
\end{align}
where $\boldsymbol{\Theta}(\omega)=\begin{bmatrix}
	\cos(\omega) & -\sin(\omega)\\
	\sin(\omega) & \cos(\omega)
\end{bmatrix}$ is a rotation matrix. Their relation is derived by subtracting the above two as
\begin{align}\label{Eq:Subtract_z}
	\boldsymbol{z}_1^{(1)*}-\boldsymbol{z}_1^{(2)*}=
	\boldsymbol{\Theta}(-\omega)
	\Big\{\big(\boldsymbol{p}_1^{(1)}(\omega)-\boldsymbol{p}_1^{(2)}(\omega)\big)-\big(\boldsymbol{p}_{\text{AP}}^{(1)}-\boldsymbol{p}_{\text{AP}}^{(2)}\big)\Big\}.
\end{align}
Noting that $\boldsymbol{p}_1^{(1)}(\omega)=\boldsymbol{p}_1^{(2)}(\omega)$ if $\omega=\omega^*$, the above is reduced as
\begin{align} \label{Eq:Subtract_z_optimal}
	{\boldsymbol{z}_1^{(1)*}-\boldsymbol{z}_1^{(2)*}}=
	\boldsymbol{\Theta}(-\omega^*)
	(\boldsymbol{p}_{\text{AP}}^{(2)}-\boldsymbol{p}_{\text{AP}}^{(1)}),
\end{align}
where $\omega^*$ is unique since the rotation matrix is not ambiguous between $[0, 2\pi)$. 
Next, the relation between two relative locations at $n=2$, say $\boldsymbol{z}_2^{(1)*}$ and $\boldsymbol{z}_2^{(2)*}$, is given as 
\begin{align} \label{Eq:IncreasingPoint}
	\boldsymbol{z}_{2}^{(1)*}-\boldsymbol{z}_{2}^{(2)*} &= \Big(\boldsymbol{z}_1^{(1)}+d \begin{bmatrix}
		\cos (\theta_1), \sin (\theta_1)
	\end{bmatrix}^T\Big) \\ &\quad \,\, - \Big(\boldsymbol{z}_1^{(2)}+d \begin{bmatrix}
		\cos (\theta_1),  \sin (\theta_1)
	\end{bmatrix}^T\Big)\nonumber\\ 
	&=\boldsymbol{z}_1^{(1)*}-\boldsymbol{z}_1^{(2)*}.
\end{align}
The above is straightforwardly extended into other relative locations at $n\in\mathbb{N}$. 	
In other words, two relative trajectories $\mathcal{Z}^{(1)}=\{\boldsymbol{z}_n^{(1)*}\}$ and $\mathcal{Z}^{(2)}=\{\boldsymbol{z}_n^{(2)*}\}$, are perfectly aligned when $\omega=\omega^*$. We complete the proof. 

\subsubsection{Linear Mobility}

\begin{figure}[t]
	\centering
	\subfigure[]{\includegraphics[width=6cm]{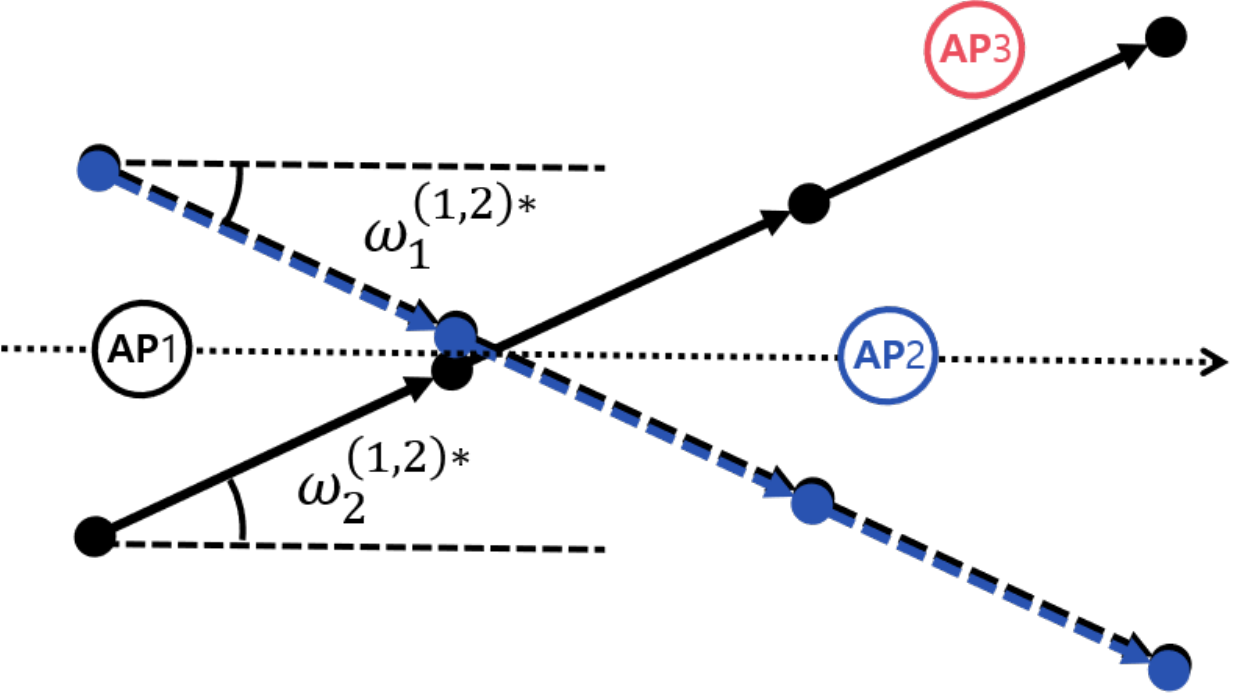}}
	\subfigure[]{\includegraphics[width=6cm]{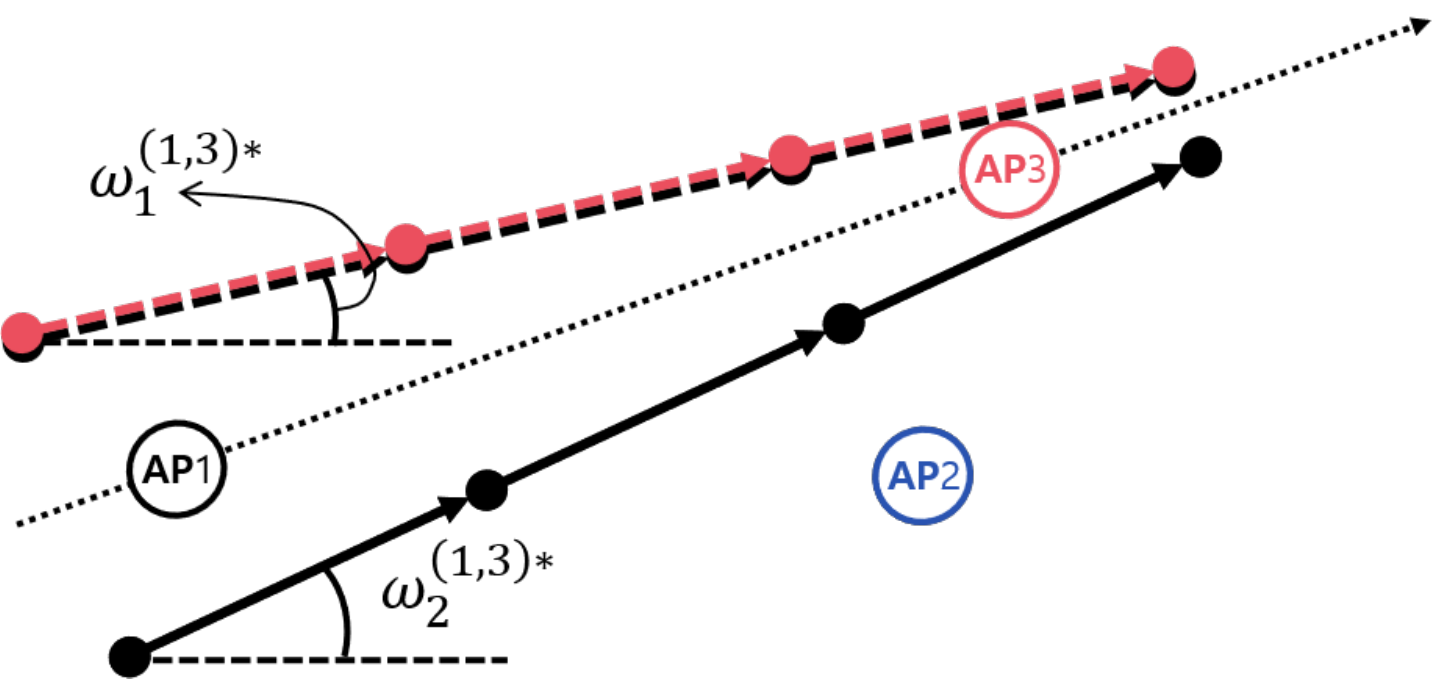}}
	\caption{The black solid line and dotted lines represent the real trajectory and fake trajectories, respectively. In this case, $\omega_2^{(1,2)*}=\omega_2^{(1,3)*}=\omega^*$ is the real one.}
	\label{Fig:proofLinear}
\end{figure}

Consider the case with APs $1$, $2$, and $3$. Recall that there are two relative trajectories for each AP, 
say $\{\mathcal{Z}^{(m)}_+, \mathcal{Z}^{(m)}_-\}_{m=1}^3$. Following the same step in the arbitrary mobility counterpart, two possible global coordinates are made for APs $1$ and $2$, denoted by $\mathcal{Z}^{(1,2)}_1$ and $\mathcal{Z}^{(1,2)}_2$, which are symmetric concerning the line between the locations of APs $1$ and $2$ (See Fig. \ref{Fig:proofLinear}). Specifically, the corresponding heading directions, denoted by $\omega_1^{(1,2)*}$ and $\omega_2^{(1,2)*}$, has the following geometric relation as
\begin{align}
	\omega_1^{(1,2)*}+\omega_2^{(1,2)*}=2 \angle \left(\boldsymbol{p}_{\text{AP}}^{(1)}-\boldsymbol{p}_{\text{AP}}^{(2)}\right),
\end{align}
where $\angle(\boldsymbol{x})$ returns the angle of the vector $\boldsymbol{x}$. Between the two, one is real whereas the other is fake. Similarly, the resultant heading directions for APs $1$ and $3$, denoted by $\omega_1^{(1,3)*}$ and $\omega_2^{(1,3)*}$, gives
\begin{align}
	\omega_1^{(1,3)*}+\omega_2^{(1,3)*}=2 \angle \left(\boldsymbol{p}_{\text{AP}}^{(1)}-\boldsymbol{p}_{\text{AP}}^{(3)}\right).
\end{align}
Using the above two equations, it is easy to identify which ones are real. 
For example, assume $\omega_2^{(1,2)*}$ and $\omega_2^{(1,3)*}$ are real, namely,  $\omega^*=\omega_2^{(1,2)*}=\omega_2^{(1,3)*}$. Unless $\angle \left(\boldsymbol{p}_{\text{AP}}^{(1)}-\boldsymbol{p}_{\text{AP}}^{(2)}\right)= \angle \left(\boldsymbol{p}_{\text{AP}}^{(1)}-\boldsymbol{p}_{\text{AP}}^{(3)}\right)$, it is obvious to identify $\omega_1^{(1,2)*}$ and $\omega_1^{(1,3)*}$ are fake since they are different. Note that if all APs exist on a straight line, it cannot be distingished between real and fake ones because symmetry is maintained, completing the proof.  

\subsection{Proof of Proposition \ref{Proposition4}} \label{Appen:Propostion4}
It can easily be verify by the fact ${\boldsymbol{z}_1^{(1)*}-\boldsymbol{z}_1^{(2)*}}=
\boldsymbol{\Theta}(-\omega^*)
(\boldsymbol{p}_{\text{AP}}^{(2)}-\boldsymbol{p}_{\text{AP}}^{(1)})$ as in \eqref{Eq:Subtract_z_optimal}. By putting norm on both sides,
\begin{align}
	\norm{\boldsymbol{z}_1^{(1)*}-\boldsymbol{z}_1^{(2)*}}=\norm{\boldsymbol{p}_{\text{AP}}^{(2)}-\boldsymbol{p}_{\text{AP}}^{(1)}}.
\end{align}
Also, $\boldsymbol{z}_{2}^{(1)*}-\boldsymbol{z}_{2}^{(2)*}=\boldsymbol{z}_1^{(1)*}-\boldsymbol{z}_1^{(2)*}$ as proved in \eqref{Eq:IncreasingPoint}. The above manners are extended into other APs position at $m_1,m_2\in\mathbb{M}$ and other relative positions at $n\in\mathbb{N}$. This finishes the proof. 

\def\bibfont{\footnotesize}

\begin{IEEEbiography}
	[{\includegraphics[width=1in,height=1.25in,clip,keepaspectratio]{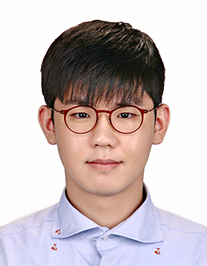}}]
	{Kyuwon Han} (S’21)
	 received the B.S degrees in physics from Chung-Ang University, South Korea in 2018. He is currently a Ph.D student in Electrical and Electronic Engineering at Yonsei University, South Korea. His research interests include indoor positioning, spectrum sharing, and machine learning techniques for wireless communications. Mr. Han received a bronze prize at IEEE Seoul Section Student Papar Award in 2019.
\end{IEEEbiography}

\begin{IEEEbiography}
	[{\includegraphics[width=1in,height=1.25in,clip,keepaspectratio]{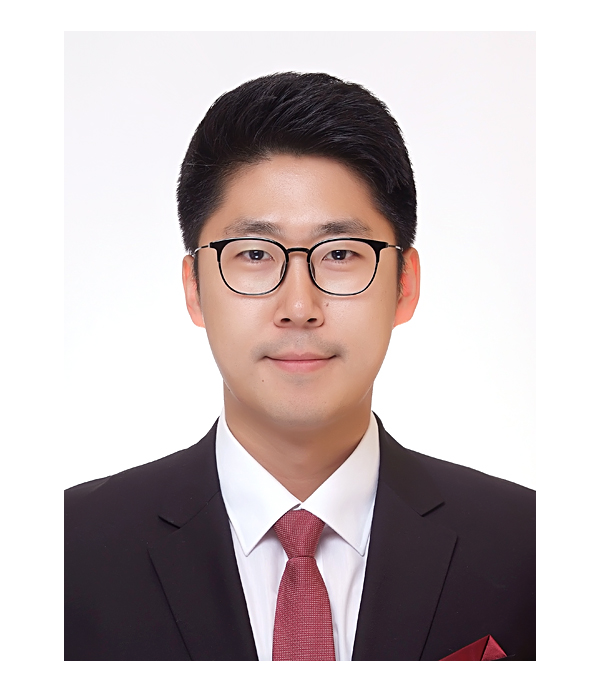}}]
	{Seung Min Yu}
	 received the B.S. (Hons.) and Ph.D. degrees in electrical and electronic engineering from Yonsei University, South Korea, in 2009 and 2013, respectively. 
	He is currently a Senior Researcher with the Korea Railroad Research Institute, South Korea. He was a Senior Engineer with Networks Business, Samsung Electronics Company, Ltd., South Korea. He was also a Researcher with the Internet Research Division, NAVER Corporation (South Korea's most popular search engine), South Korea. His research interests include indoor positioning, optimization for wireless networks, and economics of wireless systems.
	Dr. Yu received the Student Travel Grant Award at the IEEE Global Communications Conference in 2011 and the Best Paper Award at the IEEE Vehicular Technology Conference in 2013.
\end{IEEEbiography}

\begin{IEEEbiography}
	[{\includegraphics[width=1in,height=1.25in,clip,keepaspectratio]{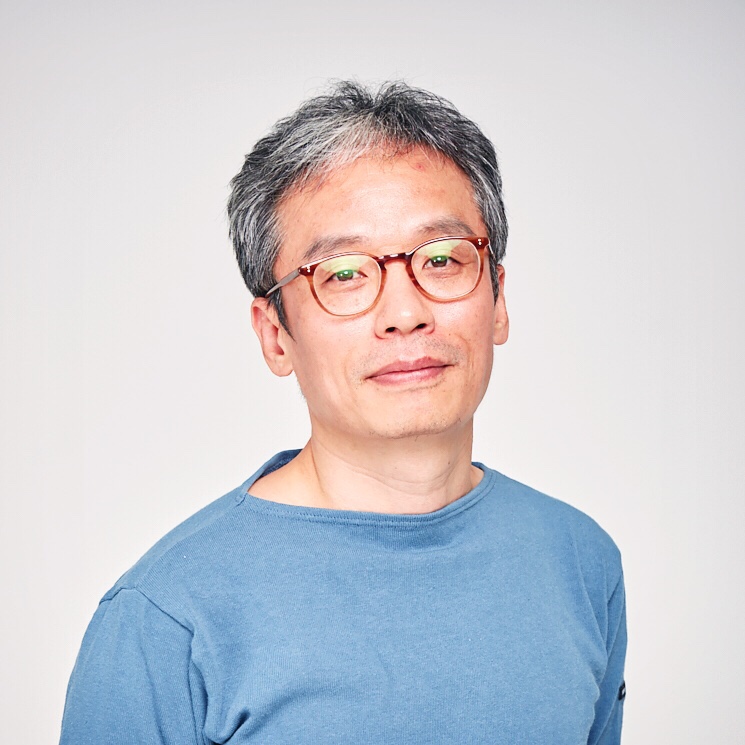}}]
	{Seung-Lyun Kim}
	is a Professor and Head of the School of Electrical \& Electronic Engineering, Yonsei University, Seoul, Korea, heading the Robotic \& Mobile Networks Laboratory (RAMO) and the Center for Flexible Radio (CFR+). He is co-directing H2020 EUK PriMO-5G project, and leading Smart Factory Committee of 5G Forum, Korea. He was an Assistant Professor of Radio Communication Systems at the Department of Signals, Sensors \& Systems, Royal Institute of Technology (KTH), Stockholm, Sweden. He was a Visiting Professor at the Control Engineering Group, Helsinki University of Technology (now Aalto), Finland, the KTH Center for Wireless Systems, and the Graduate School of Informatics, Kyoto University, Japan. He served as a technical committee member or a chair for various conferences, and an editorial board member of IEEE Transactions on Vehicular Technology, IEEE Communications Letters, Elsevier Control Engineering Practice, Elsevier ICT Express, and Journal of Communications and Network. He served as the leading guest editor of IEEE Wireless Communications and IEEE Network for wireless communications in networked robotics, and IEEE Journal on Selected Areas in Communications. He also consulted various companies in the area of wireless systems both in Korea and abroad. His research interest includes radio resource management, information theory in wireless networks, collective intelligence, and robotic networks. He published numerous papers, including the co-authored book (with Prof. Jens Zander), Radio Resource Management for Wireless Networks. His degrees include BS in economics (Seoul National University), and MS \& PhD in operations research (with application to wireless networks, Korea Advanced Institute of Science \& Technology).
\end{IEEEbiography}

\begin{IEEEbiography}
	[{\includegraphics[width=1in,height=1.25in,clip,keepaspectratio]{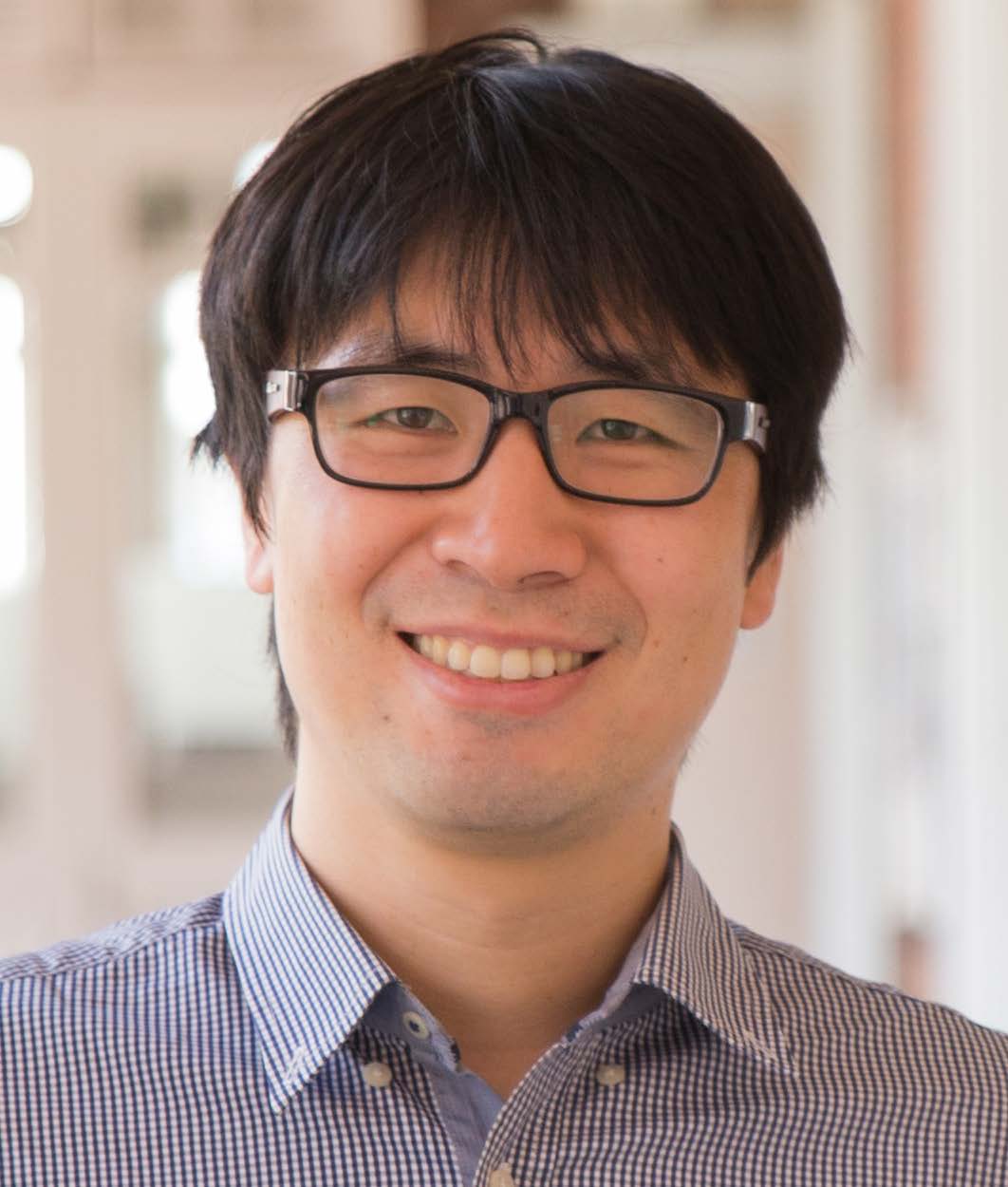}}]
	{Seung-Woo Ko}
	 (M’17) received the B.S., M.S., and Ph.D. degrees from the School of Electrical and Electronic Engineering, Yonsei University, South Korea, in 2006, 2007, and 2013, respectively. Since March 2019, he has been an Assistant Professor with the Division of Electronics and Electrical Information Engineering, Korea Maritime and Ocean University (KMOU). Before joining KMOU, he has been a Senior Researcher with LG Electronics, South Korea, from March 2013 to June 2014, and a Postdoctoral Researcher at Yonsei University, South Korea, from July 2014 to March 2016, and The University of Hong Kong, from April 2016 to February 2019. His research interests focus on intelligent wireless communications and networking, with special emphasis on edge computing and learning, vehicular technologies, and localization.
\end{IEEEbiography}

\end{document}